\pdfoutput=1
\documentclass{vldb}
\usepackage{graphicx}
\usepackage{balance}  
\usepackage{algorithm}
\usepackage[noend]{algpseudocode}
\usepackage{multirow}
\usepackage{array}
\usepackage{mathtools}
\usepackage{relsize}
\usepackage{color}
\usepackage{xcolor}
\usepackage{subcaption}
\usepackage{eqparbox}
\usepackage{caption} 
\newcolumntype{R}[1]{@{\hspace{0.5mm}}>{\raggedleft\arraybackslash}p{#1}@{\hspace{1mm}}}

\DeclareMathOperator*{\argmin}{argmin}

\newcommand\LineComment[1]{%
\textcolor{black!60}{/*\ #1\ */}
}

\newcommand\myalgo[1]{\mbox{\textsc{\MakeLowercase{#1}}}}
\newcommand\myinput[1]{\mbox{\textsc{#1}}}

\begin{document}


\title{Bitmap Filter: Speeding up Exact Set Similarity Joins\\with Bitwise Operations}



\numberofauthors{3} 

\author{
%
%
\alignauthor
Edans F. O. Sandes\\
       \affaddr{University of Brasilia}\\
       \email{edans@unb.br}
\alignauthor
George L. M. Teodoro\\
       \affaddr{University of Brasilia}\\
       \email{teodoro@unb.br}
\alignauthor
Alba C. M. A. Melo\\
       \affaddr{University of Brasilia}\\
       \email{alves@unb.br}
}

\date{30 July 1999}

\maketitle

\begin{abstract}

The Exact Set Similarity Join problem aims to find all similar sets between two collections of sets, with respect to a threshold and a similarity function such as overlap, Jaccard, dice or cosine. The na\"ive approach verifies all pairs of sets and it is often considered impractical due the high number of combinations. So, Exact Set Similarity Join algorithms are usually based on the Filter-Verification Framework, that applies a series of filters to reduce the number of verified pairs. This paper presents a new filtering technique called Bitmap Filter, which is able to accelerate state-of-the-art algorithms for the exact Set Similarity Join problem. 
The Bitmap Filter uses hash functions to create bitmaps of fixed $b$ bits, representing characteristics of the sets. Then, it applies bitwise operations (such as \textit{xor} and population count) on the bitmaps in order to infer a similarity upper bound for each pair of sets. If the upper bound is below a given similarity threshold, the pair of sets is pruned. The Bitmap Filter benefits from the fact that bitwise operations are efficiently implemented by many modern general-purpose processors and it was easily applied to four state-of-the-art algorithms implemented in CPU: AllPairs, PPJoin, AdaptJoin and GroupJoin. Furthermore, we propose a Graphic Processor Unit (GPU) algorithm based on the na\"ive approach but using the Bitmap Filter to speedup the computation. The experiments considered 9 collections containing from 100 thousands up to 10 million sets and the joins were made using Jaccard thresholds from 0.50 to 0.95. The Bitmap Filter was able to improve 90\% of the experiments in CPU, with speedups of up to 4.50$\times$ and 1.43$\times$ on average.
Using the GPU algorithm, the experiments were able to speedup the original CPU algorithms by up to 577$\times$ using an Nvidia Geforce GTX 980 Ti. 

\end{abstract}

\section{Introduction}

Data analysts often need to identify similar records stored in multiple data collections. This operation is very common for data cleaning \cite{Chaudhuri:2006}, element clustering \cite{Broder:1997:Clustering} and record linkage \cite{Cohen:2000:DIU:352595.352598,Vernica:2010:MapReduce}. In some scenarios the similarity analysis aims to detect near duplicate records \cite{PPJoin:2011}, in which slightly different data representations may be caused by input errors, misspelling or use of synonyms \cite{PartEnum:2006}. In other scenarios, the goal is to find patterns or common behaviors between real entities, such as purchase patterns \cite{mann2016} and similar user interests \cite{Spertus:2005:ESM:1081870.1081956}. 

Set Similarity Join is the operation that identifies all similar sets between two collections of sets (or inside the same collection) \cite{mann2016,Jiang:2014}. In this case, each record is represented as a set of elements (tokens). For example, the text record \textit{``Data Warehouse''} may be represented as a set of two words \{\textit{Data}, \textit{Warehouse}\} or a set of bigrams \{\textit{Da}, \textit{at}, \textit{ta}, \textit{a\textvisiblespace}, \textit{\textvisiblespace{W}}, \textit{Wa}, \textit{ar}, \textit{re}, \textit{eh}, \textit{ho}, \textit{ou}, \textit{us}, \textit{se}\}. Two records are considered similar according to a similarity function, such as Overlap, Dice, Cosine or Jaccard \cite{Augsten:2013:SJR:2601748,Theobald:2008:SRE:1390334.1390431}. For instance, two records may be considered similar if the overlap (intersection) of their bigrams is greater or equal than 6, such as in ``Databases'' and ``Databazes''.

Many different solutions have been proposed in the literature for Set Similarity Join \cite{GramCount:2001,Chaudhuri:2006,AllPairs:2007,PPJoin:2011,AdaptJoin:2012,GroupJoin:2012,Gionis:1999:LSH,Zhai:2011:LSH,Satuluri:2012:LSH,Chakrabarti:2015:LSH,Yu:2017:LSH, SOHRABI20171:LSH,Zhang:2017:Exact}. With respect to the guarantee that all similar pairs will be found, these solutions are divided into Exact approaches \cite{GramCount:2001,Chaudhuri:2006,AllPairs:2007,PPJoin:2011,AdaptJoin:2012,GroupJoin:2012,Zhang:2017:Exact} and approximate aproaches \cite{Gionis:1999:LSH,Zhai:2011:LSH,Satuluri:2012:LSH,Chakrabarti:2015:LSH,Yu:2017:LSH, SOHRABI20171:LSH}. The approximate algorithms usually rely on the Local Sensitive Hashing technique (LSH) \cite{LSH:1998} and they are very competitive, but in the scope of this paper only the exact approaches will be considered.

The solutions for Exact Set Similarity Join are usually based on the Filter-Verification Framework \cite{mann2016, Wang:2017:LSR:3099622.3099624}, which is divided into two stages: a) the candidate generation uses a series of filters to produce a reduced list of candidate pairs; b) the verification stage verifies each candidate pair in order to check which ones are similar, with respect to the selected similarity function and threshold. 

The filtering stage is commonly based on Prefix Filter \cite{Chaudhuri:2006,AllPairs:2007} and Length Filter\cite{GramCount:2001}, which are able to prune a considerable number of candidate pairs without compromising the exact result. Prefix Filter\cite{Chaudhuri:2006,AllPairs:2007} is based on the idea that two similar sets must share at least one element in their initial elements, given a defined element ordering. Length Filter \cite{GramCount:2001} prunes candidate lists according to their lengths, considering that elements with a big difference in length will have low similarity. Other filters have also been proposed, such as the Positional Filter \cite{PPJoin:2011}, which is able to prune candidate pairs based on the position where the token occurs in the set, and the Suffix Filter \cite{PPJoin:2011}, that prunes pairs by applying a binary search on the last elements of the sets. 

There are many algorithms in the literature combining the aforementioned filters. AllPairs\cite{AllPairs:2007} uses the Prefix and Length filters, PPJoin \cite{PPJoin:2011} extends AllPairs with the Positional Filter, PPJoin+ \cite{PPJoin:2011} extends PPJoin with the Suffix Filter. AdaptJoin \cite{AdaptJoin:2012} extends PPJoin using a variable schema with adaptive prefix lengths and GroupJoin \cite{GroupJoin:2012} applies the filters in groups of similar sets, avoiding redundant computations. 
These algorithms usually rely on the assumption that the verification stage contributes significantly to the overall running time. However, in \cite{mann2016} a new verification procedure using an early termination condition was proposed that significantly reduced the execution time of the verification phase. With this, the proportion of execution time in the candidate generation phase increased and, as a consequence, 
the simplest filtering schemes are now able to produce the best performance. In \cite{mann2016}, the best results were obtained by AllPairs, PPJoin, GroupJoin and AdaptJoin, whereas AllPairs was the best algorithm on the largest number of experiments \cite{mann2016}. On the contrary, PPJoin+\cite{PPJoin:2011}, which uses a sophisticated suffix filter technique, was the slowest algorithm in \cite{mann2016}.

The state-of-the-art set similarity join algorithms still need to be improved in order to allow a faster join, otherwise very large collections may not be processed in a reasonable time. Based on the recent findings on the filter-verification trade-off overhead \cite{mann2016}, the authors claimed that future filtering techniques should invest in fast and simple candidate generation methods instead of sophisticated filter-based algorithms, otherwise the filters effectiveness will not pay off their overhead  when compared to the verification procedure. 

The main contribution of this paper is a new low overhead filter called Bitmap Filter, which is able to efficiently prune candidate pairs and improve the performance of many state-of-the-art exact set similarity join algorithms such as AllPairs\cite{AllPairs:2007}, PPJoin\cite{PPJoin:2011}, AdaptJoin\cite{AdaptJoin:2012} and GroupJoin \cite{GroupJoin:2012}. Using bitwise operations implemented by most modern processors, the Bitmap Filter is able to speedup the performance of AllPairs, PPJoin, AdaptJoin and GroupJoin algorithms by up to 4.50$\times$ (1.43$\times$ on average) considering 9 different collections and Jaccard threshold varying from $0.50$ to $0.95$. As far as we know, there is no other filter for the Exact Set Similarity Join problem that is based on efficient bitwise operations. 

The Bitmap Filter is sufficiently flexible to be used in the candidate generation or verification stages. Furthermore, it can be efficiently implemented in devices such as Graphical Processor Units (GPU). As a secondary contribution of this paper, we thus propose a GPU implementation of the Bitmap Filter that is able to speedup the join by up to $577\times$ using a Nvidia Geforce GTX 980 Ti card, considering 6 collections of up to 606,770 sets. 

The rest of this paper is structured as follows. Section \ref{sec:background} presents the background of the Set Similarity Join problem. Then, Section \ref{sec:bitmap_filter} proposes the Bitmap Filter and explains it in detail. Section~\ref{sec:implementations} shows the design of the proposed CPU and GPU implementations. Section \ref{sec:experimental_results} presents the experimental results and Section \ref{sec:conclusion} concludes the paper.

\section{Background}
\label{sec:background}

In this section, we discuss the Set Similarity Join problem. First, the problem is formalized (Section \ref{sec:problem}). Then, the Filter-Verification Framework is presented (Section \ref{sec:related_work}). After that, some commonly used filters employed in the Filter-Verification Framework are explained (Section \ref{sec:filter}). Finally, state-of-the-art algorithms are described (Section \ref{sec:algorithms}).

\subsection{Problem Definition}
\label{sec:problem}

Given two collections $R=\{r_1,r_2,\cdots,r_{|R|}\}$ and \linebreak$S=\{s_1,s_2,\cdots,r_{|S|}\}$, where sets $r_i$ and $s_j$ are made of tokens from $T$ ($r_i\subseteq T$ and $s_j\subseteq T$), the Set Similarity Join problem aims to find all pairs $(r_i, s_j)$ from $R$ and $S$ that are considered similar with respect to a given similarity function $sim$, i.e. $sim(r_i, s_j)\ge \tau$, where $\tau$ is a user defined threshold \cite{mann2016,PPJoin:2011}. When R and S are the same collection, the problem is defined as a self-join, otherwise it is called $RS$-join. The set similarity join operation can be written as presented in Equation \ref{eq:ssjoin} \cite{mann2016}.

\begin{equation}
\label{eq:ssjoin}
R~\tilde{\bowtie}~S = \{(r_i, s_j)\in R\times S | sim(r_i, s_j)\ge \tau\}
\end{equation}

The most used similarity functions for comparing sets are the Overlap, Jaccard, Cosine and Dice \cite{Augsten:2013:SJR:2601748,Theobald:2008:SRE:1390334.1390431}, as presented in Table \ref{tb:simfuncs}. The Overlap function returns the number of elements in the intersection of sets $r$ and $s$ (we suppress the subscripts $i$ and $j$ when the context is free of ambiguity). Jaccard, Cosine and Dice are similarity functions that normalize the overlap by dividing the intersection $|r \cap s|$ by the union $|r \cup s|$ (Jaccard), or by a combination of the lengths $|r|$ and $|s|$ (Cosine and Dice). The similarity functions can be converted to each other using the threshold equivalence presented in Table \ref{tb:simfuncs} \cite{Jiang:2014, Ribeiro:2011}, where $\tau$, $\tau_j$, $\tau_c$ and $\tau_d$ are the overlap, Jaccard, cosine and dice thresholds, respectively.

\begin{table}
\centering
\caption{Similarity functions and their equivalence to overlap\protect\cite{Jiang:2014, Ribeiro:2011}\label{tb:simfuncs}}
\begin{tabular}{|c|c|c|}
\hline
Sim. Function & $sim(r,s)$ & Equiv. Overlap \\
\hline
Overlap & $|r\cap s|$ & $\tau$\\[1.1ex]
\hline
Jaccard & $\frac{|r\cap s|}{|r\cup s|}$ & $\tau = \frac{\tau_j}{1+\tau_j}(|r|+|s|)$\\[1.1ex]
\hline
Cosine & $\frac{|r\cap s|}{\sqrt{|r|\cdot|s|}}$ & $\tau = \tau_c\sqrt{|r|\cdot|s|}$\\[1.1ex]
\hline
Dice & $\frac{2|r\cap s|}{|r|+|s|}$ & $\tau = \tau_d\frac{|r|+|s|}{2}$\\[1.1ex]
\hline
\end{tabular}
\end{table}

To solve the set similarity join problem, the na\"ive algorithm (Algorithm \ref{algo:naive}) compares all pairs from collections $R$ and $S$ and verifies if the similarity of each pair (with respect to the similarity function) is greater or equal than the desired threshold. Since this method verifies $|R|\times|S|$ elements, it does not scale to a large number of sets.

\begin{algorithm}
\caption{Na\"ive Algorithm \label{algo:naive}}
\begin{algorithmic}[1]
\Statex \textbf{Input:} set collections $R$ and $S$; threshold $\tau$
\Statex \textbf{Output:} $pairs=R~\tilde{\bowtie}~S$ of similar sets
\For{$r \in R$}
\For{$s \in S$}
\If {$verify(r,s,\tau)$}
\State $pairs \leftarrow pairs \cup (r,s)$
\EndIf
\EndFor
\EndFor
\State \Return $pairs$ \label{algo:basic:return}
\end{algorithmic}
\end{algorithm}

\subsection{Filter-Verification Framework}
\label{sec:related_work}

Many algorithms, such as AllPairs\cite{AllPairs:2007}, PPJoin\cite{PPJoin:2011}, \linebreak PPJoin+\cite{PPJoin:2011}, AdaptJoin\cite{AdaptJoin:2012}, and GroupJoin \cite{GroupJoin:2012}, have been proposed to reduce the number of verification  operations, aiming to solve the set similarity join problem efficiently. These algorithms typically employ filtering techniques using the Filter-Verification Framework, as shown in Algorithm \ref{algo:basic}.

\begin{algorithm}
\caption{Filter-Verification Framework \label{algo:basic}}
\begin{algorithmic}[1]
\Statex \textbf{Input:} set collections $R$ and $S$; threshold $\tau$
\Statex \textbf{Output:} $pairs=R~\tilde{\bowtie}~S$ of similar sets
\Statex \LineComment{Initialization}
\State $I \leftarrow index(S)$ \label{algo:basic:index}
\For{$r \in R$}
	\Statex {\hskip\algorithmicindent}\LineComment{Candidate Generation Stage}
	\State $candidate \leftarrow \{\}$
	\For{$t \in filter_1(r,\tau)$} \label{algo:basic:filter_1}
		\For{$s \in I[t]$}
			\If {$\textbf{not} filter_2(r,s,\tau)$} \label{algo:basic:filter_2}
				\State $candidate \leftarrow candidate \cup \{s\} $
			\EndIf
		\EndFor
	\EndFor
	\Statex {\hskip\algorithmicindent}\LineComment{Verification Stage}
	\For{$s \in candidate$}
		\If {$\textbf{not} filter_3(r,s,\tau)$} \label{algo:basic:filter_3}
			\If {$verify(r,s,\tau)$} \label{algo:basic:verify}
				\State $pairs \leftarrow pairs \cup (r,s)$ \label{algo:basic:similar_pairs}
			\EndIf
		\EndIf		
	\EndFor 
	\State \Return $pairs$ \label{algo:basic:return}
\EndFor
\end{algorithmic}
\end{algorithm}

In the initialization step, the $S$ collection is indexed (line~\ref{algo:basic:index}) with respect to the tokens found in each set $s \in S$, such that $I[t]$ stores all sets containing token $t$. Then, for each set $r \in R$, the algorithm iterates in a two phase loop: candidate generation and verification stage. 

In the candidate generation loop (lines 4-7), the algorithm iterates over tokens $t$ from set $r$ (filtered by function \emph{filter}$_1$ in line \ref{algo:basic:filter_1}) in order to find in index $I$ all sets $s \in S$ that also contain token $t$. Each set $s$ may also be filtered (function \emph{filter}$_2$ in line \ref{algo:basic:filter_2}) and, if the set is not pruned, it is inserted into a candidate list. 

The verification stage (lines 8-11) checks, for each unique candidate $r$, if the pair $(r,s)$ is similar with respect to the similarity function and threshold $\tau$ (line~\ref{algo:basic:verify}) and, if so, the pair is inserted in the similar pair list (line~\ref{algo:basic:similar_pairs}). A filter may also be applied in the verification loop (function \emph{filter}$_3$ in line~\ref{algo:basic:filter_3}) to reduce the number of verified pairs. Finally, the similar pairs are returned (line~\ref{algo:basic:return}).

\subsection{Filtering Strategies}
\label{sec:filter}
This section explains the most widely used filtering techniques in the literature. 

\begin{figure*}[hbtp]
\centering
\begin{subfigure}[t]{0.23\textwidth}
\includegraphics[width=4cm]{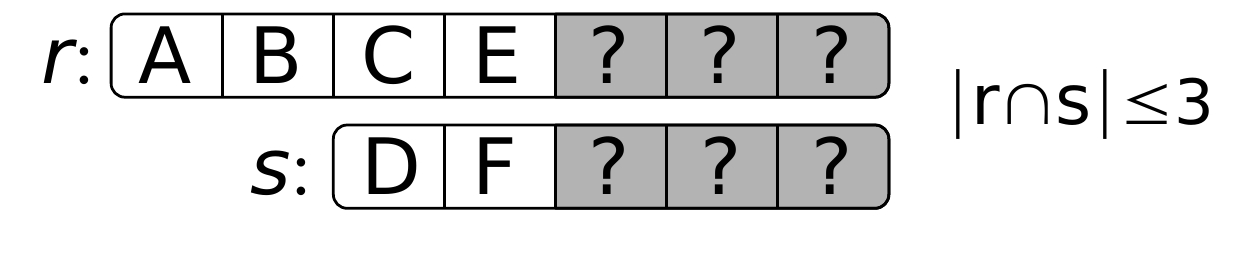}
\caption{Prefix Filter Technique\label{fig:prefixfilter}}
\end{subfigure}
\begin{subfigure}[t]{0.23\textwidth}
\includegraphics[width=4cm]{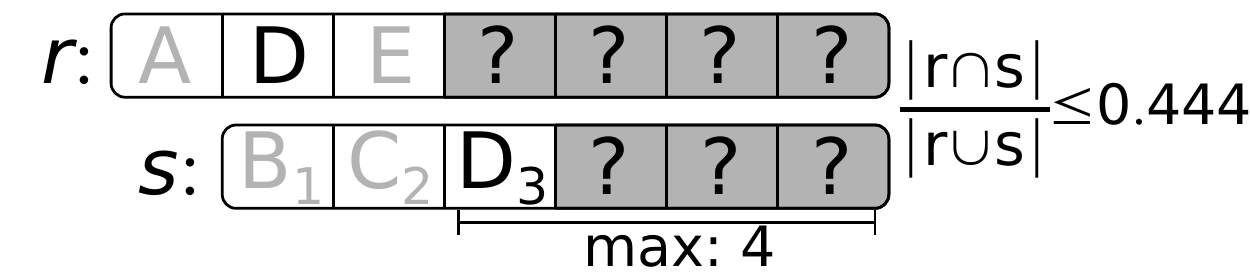}
\caption{Positional Filter Technique\label{fig:positionalfilter}}
\end{subfigure}
\begin{subfigure}[t]{0.23\textwidth}
\includegraphics[width=4cm]{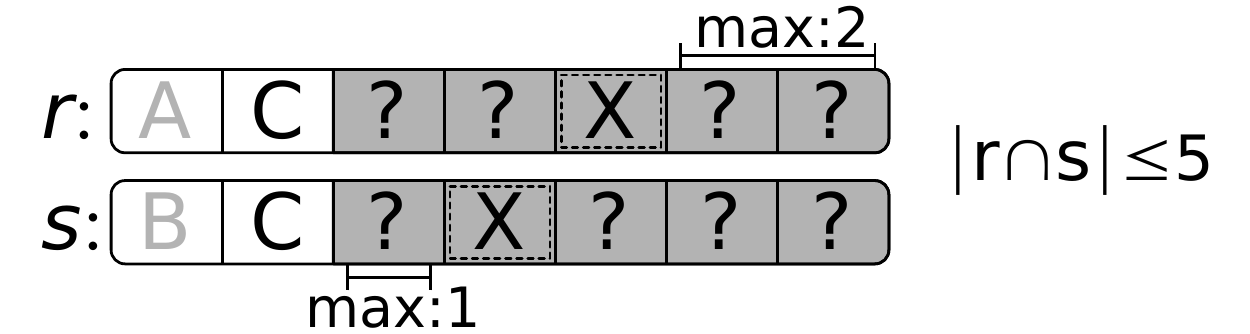}
\caption{Suffix Filter Technique\label{fig:suffixfilter}}
\end{subfigure}
\begin{subfigure}[t]{0.23\textwidth}
\includegraphics[width=4cm]{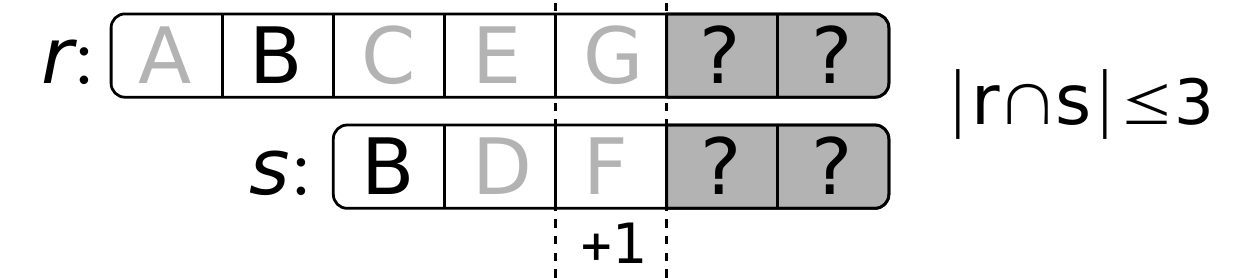}
\caption{$\ell$-Prefix Filter ($\ell=2$)\label{fig:adaptjoin}}
\end{subfigure}
\caption{Filtering strategies}
\end{figure*}

\subsubsection{Prefix Filter}
\label{sec:filter:prefix}

The Prefix Filter technique \cite{Chaudhuri:2006,AllPairs:2007} creates an index over the first tokens in each set, considering that the tokens are sorted, usually based on the token frequencies. The sizes of the selected prefixes in each set are such that if two sets $r$ and $s$ are similar (i.e. $sim(r,s)\ge \tau$), then their prefixes will have at least one token in common. For instance, Figure \ref{fig:prefixfilter} shows two sets $r$ and $s$ with sizes \mbox{$|r|=7$} and \mbox{$|s|=5$}. Considering that the overlap similarity must be greater or equal than $\tau=4$, then the prefix sizes are respectively \mbox{\emph{prefix}$(r)=4$} and \mbox{\emph{prefix}$(s)=2$}. In this scenario, if there is no overlap in the prefixes (white cells), then it will be impossible to have at least $\tau=4$ overlaps considering the remaining cells (in gray). In fact, considering the overlap similarity function, the prefix size for a set with size $|r|$ is \mbox{\emph{prefix}$(r)=|r|-\tau+1$} \cite{AllPairs:2007}. The prefix sizes for other similarity functions are presented in Table \ref{tb:lengthfilter} \cite{Jiang:2014}.


After the prefix is created, the algorithm produces a list of candidate pairs for each set $r$. This list is composed of all the sets $s\in S$ that are stored in the index pointed by the elements in \emph{prefix}$(r)$. The candidate pairs are then verified using the similarity function and, if the threshold $\tau$ is attained, this pair is marked as similar.

\subsubsection{Length Filter}
\label{sec:filter:length}

Length Filter \cite{GramCount:2001} benefits from the idea that two sets $r$ and $s$ with very different lengths tend to be dissimilar. Given the size of set $|r|$ and the threshold $\tau$, Table \ref{tb:lengthfilter} presents the upper and lower bounds for the size $|s|$ such that, outside of these bounds, the similarity $sim(r,s)$ is surely lower than the threshold. Table \ref{tb:lengthfilter} presents these bounds considering the overlap, Jaccard, cosine and dice similarity functions \cite{mann2016}. 

\begin{table}
\centering
\caption{Filter Parameters per Similarity Function\label{tb:lengthfilter}}
\setlength\tabcolsep{4.0pt}
\begin{tabular}{|c|r@{\hspace{2mm}$\le|s|\le$\hspace{2mm}}l|c|}
\hline
Function & \multicolumn{2}{c|}{Lower/Upper Bounds} & Prefix Length \\
\hline
Overlap & $\tau$ & $\infty$  & $|r|-\tau+1$ \\[1.1ex]
\hline
Jaccard & $|r|\tau_j$ & $\frac{|r|}{\tau_j}$& $\left\lfloor (1-\tau_j)|r| \right\rfloor + 1$ \\[1.1ex]
\hline
Cosine & $|r|\tau_c^2$ & $\frac{|r|}{\tau_c^2}$& $\left\lfloor (1-\tau_c^2)|r| \right\rfloor + 1$ \\[1.1ex]
\hline
Dice & $|r|\frac{\tau_d}{2-\tau_d}$ & $\frac{2-\tau_d}{\tau_d}|r|$ & $\left\lfloor (1-\frac{\tau_d}{2-\tau_d})|r| \right\rfloor + 1$\\[1.1ex]
\hline
\end{tabular}
\end{table}

Using this filter, the algorithms can ignore a set whose size is outside the defined bounds. Considering that the sets in the list are ordered by their lengths, whenever a set from an index list contains more elements than the defined upper bound, all the remaining elements of that index can be ignored.

\subsubsection{Positional Filter}
\label{sec:filter:positional}

Positional Filter \cite{PPJoin:2011} tracks the position where a given token appears in the sets. Using this position, it is possible to determine a tighter upper bound for the similarity metric, and this tends to reduce the number of candidate pairs. For example, Figure \ref{fig:positionalfilter} presents two sets $s$ and $r$ with sizes $r=7$ and $s=6$. A Jaccard threshold $\tau_j=0.6$ will produce prefixes with size $3$ in both sets. In this figure, the elements in \emph{prefix}$(s)$ are presented with their position in the set ($B_1$, $C_2$ and $D_3$). The first matching token ``D'' occurs in position 3 in set $s$ and, with this information and the sizes $|r|$ and $|s|$, it can be inferred that the maximum possible overlap is $4$. Using the same idea, the minimum possible union is $9$ and, thus, the maximum Jaccard coefficient is $\frac{4}{9}=0.444$. Since the upper Jaccard coefficient is below the threshold $\tau_j=0.6$, this candidate can be pruned. The same idea can be applied to cosine and dice similarity functions.


\subsubsection{Suffix Filter}
\label{sec:filter:suffix}

The Suffix Filter \cite{PPJoin:2011} takes one probe token $t_p$ from the suffix of set $r$ and applies a binary search algorithm over set $s$ in order to divide these sets in two parts (left/right). Since the tokens in each set are sorted, tokens in the left side of one set cannot overlap with tokens in the right side of the other set. Using this idea, it can be inferred a tighter upper bound for the overlap between two sets. For example, in Figure \ref{fig:suffixfilter} there are two sets $r$ and $s$ with sizes $|r|=|s|=7$, where the required overlap threshold between the sets is $\tau=6$. In the figure, prefixes are represented by white cells and the suffixes by gray cells. Comparing the tokens in the prefixes, it is possible to find one match (token ``C''), but 5 additional matches still need to be found. Then, the Suffix Filter chooses the middle token (``X'') from the suffix of set $r$ and it seeks for the same token in set $s$. In the figure, the token ``X'' divides the set $r$ into two subsets with two elements each, and the set $s$ is divided into two subsets, with one element to the left and three elements to the right. Considering that the elements are sorted, the maximum number of additional overlaps in each side is one element in the left of ``X'' and two elements in the right. This would lead to an upper bound of 5, which is less than the required overlap ($\tau=6$). Thus, this pair can be pruned.


\subsubsection{Variable-Length Prefix Filter}
\label{sec:filter:variable_prefix}

The original Prefix Filter relies on the idea that the prefixes of sets $r$ and $s$ must share at least one token in order to satisfy the threshold. The Variable-Length Prefix Filter \cite{AdaptJoin:2012} extends this idea using adaptive prefix lengths. For example, using a $2$-length prefix schema, it is required at least $2$ common tokens in the prefixes. The prefix size for a set with size $|r|$ in a $\ell$-prefix schema is \mbox{\emph{prefix}$_\ell(r)=|r|-\tau+\ell$}. Figure \ref{fig:adaptjoin} shows an example with sets $r$ and $s$ with sizes $|r|=7$ and $|s|=5$. Considering a required overlap $\tau=4$,  the $2$-prefix schema chooses \mbox{\emph{prefix}$_2(r)=5$} and \mbox{\emph{prefix}$_2(s)=3$}. Since the prefixes (white cells) present less than two overlaps, it is impossible to satisfy the threshold $\tau=4$, so the pair is pruned. 


\subsection{Algorithms}
\label{sec:algorithms}

 This section presents an overview of state-of-the art algorithms that use the Filter-Verification Framework (Algorithm~\ref{algo:basic}). Table \ref{tb:algorithms} shows five algorithms, related acronyms and filters used by the algorithm. In this paper, we will focus on the best four algorithms as stated by \cite{mann2016}: AllPairs, PPJoin, AdaptJoin, and GroupJoin. 

\begin{table}
\centering
\caption{Algorithms and the selected filters. \label{tb:algorithms}}
\begin{tabular}{|@{ }c@{ }l|l|l|}
\hline
\multicolumn{2}{|c|}{\textbf{Algorithm}} & \multicolumn{1}{c|}{$\boldsymbol{filter_1}$} & \multicolumn{1}{c|}{$\boldsymbol{filter_2}$} \\
\hline
\myalgo{ALL} & AllPairs\cite{AllPairs:2007} & Prefix & Length\\
\hline
\myalgo{PPJ} & PPJoin\cite{PPJoin:2011} & Prefix & Length/Pos. \\
\hline
\myalgo{PPJ+} & PPJoin+\cite{PPJoin:2011} & Prefix & Length/Pos./Suffix \\
\hline
\myalgo{GRO} & GroupJoin\cite{GroupJoin:2012} & Prefix & Length/Pos. \\
\hline
\myalgo{ADA} & AdaptJoin\cite{AdaptJoin:2012} & $\ell$-Prefix & Length\\
\hline
\end{tabular}
\end{table}

\textbf{AllPairs}\cite{AllPairs:2007} is an algorithm that uses the Prefix Filter (Section~\ref{sec:filter:prefix}) and Length Filter (Section~\ref{sec:filter:length}). In the Filter-Verification Framework, the Prefix Filter is applied as \emph{filter}$_1$ and the Length Filter is applied as \emph{filter}$_2$.

\textbf{PPJoin}\cite{PPJoin:2011} extends the AllPairs algorithm, including the Positional Filter (Section~\ref{sec:filter:positional}). In the Filter-Verification Framework, the Prefix Filter is applied as \emph{filter}$_1$ and the Length and Positional filters are applied as \emph{filter}$_2$.

\textbf{PPJoin+}\cite{PPJoin:2011} extends the PPJoin algorithm with the Suffix Filter (Section~\ref{sec:filter:suffix}) over the candidate pairs. In the Filter-Verification Framework, the Prefix Filter is applied as \emph{filter}$_1$ and the Length, Positional and Suffix filters are applied as \emph{filter}$_2$.

\textbf{AdaptJoin}\cite{AdaptJoin:2012} is an algorithm that uses the Length Filter (Section~\ref{sec:filter:length}) and the Variable-Length Prefix Filter (Section~\ref{sec:filter:variable_prefix}). The Filter-Verification Framework is modified such that additional iterations are executed with different prefix sizes, according to the characteristics of the collection. The Variable-Length Prefix Filter is applied as \emph{filter}$_1$ and the Length Filter is applied as \emph{filter}$_2$.

\textbf{GroupJoin}\cite{GroupJoin:2012} applies the Prefix Filter (Section~\ref{sec:filter:prefix}), Length Filter (Section~\ref{sec:filter:length}), and Positional Filter (Section \ref{sec:filter:positional}) over groups of sets containing the same prefixes and sizes. In the verification stage of the Filter-Verification \linebreak Framework, the grouped candidate pairs are expanded to individual elements and all possible pairs are verified. With this grouping schema, the filters can be applied only to one representative of the group, what considerably reduces the number of comparisons for groups with many elements. Figure \ref{fig:groupjoin} presents two groups of sets containing 3 elements sharing the same prefix.

\begin{figure}[t]
\centering
\includegraphics[width=5cm]{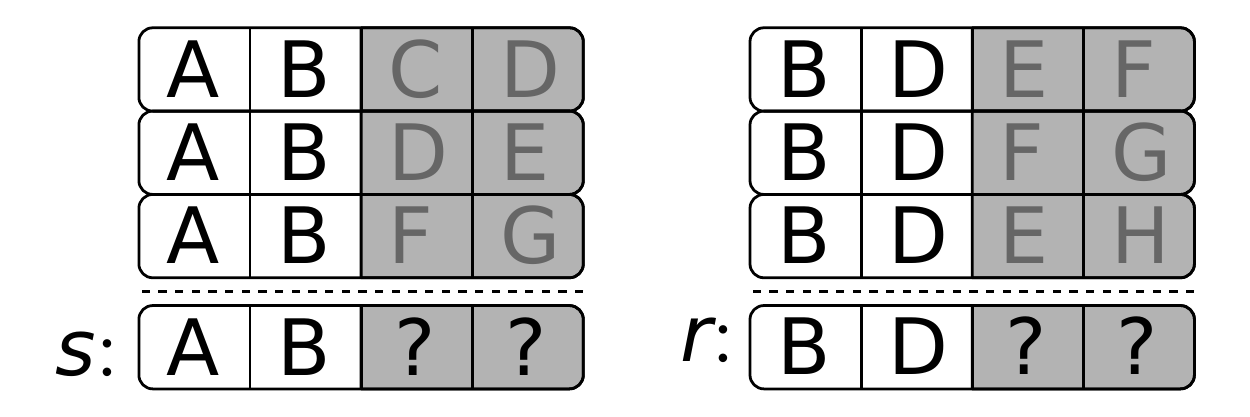}
\caption{Grouped elements in GroupJoin\label{fig:groupjoin}}
\end{figure}

\section{Bitmap Filter}
\label{sec:bitmap_filter}

In this section we propose the Bitmap Filter, which uses a new bitwise technique capable of speeding up similarity set joins without compromising the exact result. The bitwise operations are widely used in modern microprocessors and most of the compilers are capable to map them to fast hardware instructions. 

First, some preliminaries are given in Section \ref{sec:bitmap:preliminaries}. Then, Section \ref{sec:bitmap:definitions} presents how the bitmap is generated. Section \ref{sec:bitmap:upperbound} shows how the overlap upper bound can be calculated using bitmaps. Then, in Section \ref{sec:bitmap:cutoff}, a cutoff point is defined in order to determine when the Bitmap Filter is effective. Finally, the pseudocode of the Bitmap Filter is given in Section \ref{sec:bitmap:algorithm}.

\subsection{Preliminaries}
\label{sec:bitmap:preliminaries}

The similarity functions between two sets $r$ and $s$ can be calculated over their binary representation. For instance, sets $r=\{1,5,8\}$ and $s=\{3,4,5,6\}$ can be mapped to binary sets represented by $b_r=10001001$ and $b_s=00111100$, where $bit_k$ at position $k$ represents the occurrence of token $k$ in the corresponding set. In this representation, the size of the binary sets must be sufficiently large to map all possible tokens in the collection. Considering that each bit in the binary set is a dimension, collections with large number of distinct tokens will produce a high dimensionality binary sets that are usually very sparse.

Modern microprocessors implement many bitwise instructions that can be used to speedup computations of similarity functions over binary sets. One of those instructions is the \textit{population count} (popcount), which is used to count the number of ``ones'' in the binary representation of a numeric word \cite{nehalem2008inside}. For example, the overlap of the sets $r$ and $s$ mentioned in the previous paragraph can be calculated by $popcount(w_r \land w_s) = popcount(00001000) = 1$, using only two instructions in general microprocessors (one \textit{AND} instruction and one \textit{POPCNT} instruction). The same idea can be extended to Jaccard coefficient: $\frac{popcount(w_r \land w_s)}{popcount(w_r \lor w_s)}=\frac{popcount(00001000)}{popcount(10111101)}=\frac{1}{6}$. 

High dimensionality binary sets may be implemented by the concatenation of two or more fixed-length words, but as the dimensionality increases, more instructions are necessary to calculate the similarities, decreasing the performance of the computation. The Bitmap Filter proposed in this section reduces the dimensionality of the binary sets without compromising the exactness of the Set Similarity Join results. Using a binary set with reduced size, called Bitmap from now on, the bitwise instructions may be much more effective.

\subsection{Bitmap Generation}
\label{sec:bitmap:definitions}

Consider a set $s=\{s_1, s_2, \cdots, s_n\}$ composed by tokens $s_k \in T = \{1,2,3,\cdots,|T|\}$. Let $b_s$ be a bitmap with $b$ bits, representing set $s$. The bit at position $i \in [1..b]$ is represented by $b_s[i]$. Let $h(t): T \rightarrow [1..b]$ be a hash function that maps each token $t\in T$ to a specific bit position $i$. 

To create the bitmap $b_s$ from the set $s$, we define a generation function $Bitmap(s): T^* \rightarrow {0,1}^b$ that maps a subset of tokens into the bitmap $b_s$. Three different $Bitmap(s)$ generation functions are proposed in this paper: Bitmap-Set, Bitmap-Xor and Bitmap-Next. The difference between them is in how they handle collisions produced by the hash function $h(t)$.

\textbf{Bitmap-Set}: For each element $s_k\in s$, the bit at position $i=h(s_k)$ will be set. If multiple elements map to the same $i$ position, only the first element will effectively change the bit status (the other elements will leave this bit set to `one'). Algorithm \ref{algo:bitmap-set} presents the pseudocode for this function.

\begin{algorithm}
\caption{Bitmap-Set \label{algo:bitmap-set}}
\begin{algorithmic}[1]
\Function{Bitmap-Set}{$s = \{s_1,s_2,\cdots,s_n\}$}
\State $b_s \gets 00000 \cdots 0$
\For{$s_k \in s$}
\State $b_s[h(s_k)]=1$
\EndFor
\State \Return $b_s$
\EndFunction
\end{algorithmic}
\end{algorithm}

\textbf{Bitmap-Xor}: For each element $s_k\in s$, the bit at position $i=h(s_k)$ will have its value changed ($1\rightarrow 0$ or $0\rightarrow 1$). In the end, the bit $b[i]$ will remain set only if there is an odd number of elements mapped to position $i$. Algorithm \ref{algo:bitmap-xor} presents the pseudocode for this function, where $\oplus$ is the \textit{xor} operation.

\begin{algorithm}
\caption{Bitmap-Xor \label{algo:bitmap-xor}}
\begin{algorithmic}[1]
\Function{Bitmap-Xor}{$s = \{s_1,s_2,\cdots,s_n\}$}
\State $b_s \gets 00000 \cdots 0$
\For{$s_k \in s$}
\State $b_s[h(s_k)]=b[h(s_k)] \oplus 1$
\EndFor
\State \Return $b_s$
\EndFunction
\end{algorithmic}
\end{algorithm}

\textbf{Bitmap-Next}: For each element $s_k\in s$, if the bit at position $i=h(s_k)$ is unset, set it to 1, otherwise set the next unset bit after position $i$, circulating to the beginning of the bitmap if it reaches the last bit. This method ensures that there will be exactly $n=|s|$ 1's in the bitmap, unless $n>b$. In this case, all $b$ bits will be set. Algorithm \ref{algo:bitmap-next} presents the pseudocode for this method.

\begin{algorithm}
\caption{Bitmap-Next\label{algo:bitmap-next}}
\begin{algorithmic}[1]
\Function{Bitmap-Next}{$s = \{s_1,s_2,\cdots,s_n\}$}
\If{$(n \ge b)$}
	\State $b_s \gets 11111 \cdots 1$
\Else
	\State $b_s \gets 00000 \cdots 0$
	\For{$s_k \in s$}
	\State $i \gets h(s_k)$
	\While{$b_s[i]==1$}
	\State $i = (i+1) \mod b$
	\EndWhile
	\State $b_s[i]=1$
	\EndFor
\EndIf
\State \Return $b_s$
\EndFunction
\end{algorithmic}
\end{algorithm}

Figure \ref{fig:bitmapexample} shows an example of bitmap generation considering the bitmap size $b=16$ and a set $r$ with 10 tokens. The hash function $h(x)$ is represented by the arrows pointing to each bitmap slot. The three bitmaps $b_r$ were produced from Bitmap-Set, Bitmap-Xor and Bitmap-Next, respectively from top to bottom. Compared to Bitmap-Set, it can be seen that the bitmap generated by Bitmap-Xor has 1 different bit and the Bitmap-Next has 3 different bits. Each generation method will fit better in different circumstances that will be explored in Section \ref{sec:bitmap:upperbound}.

\begin{figure}[hbtp]
\centering
\includegraphics[width=5cm]{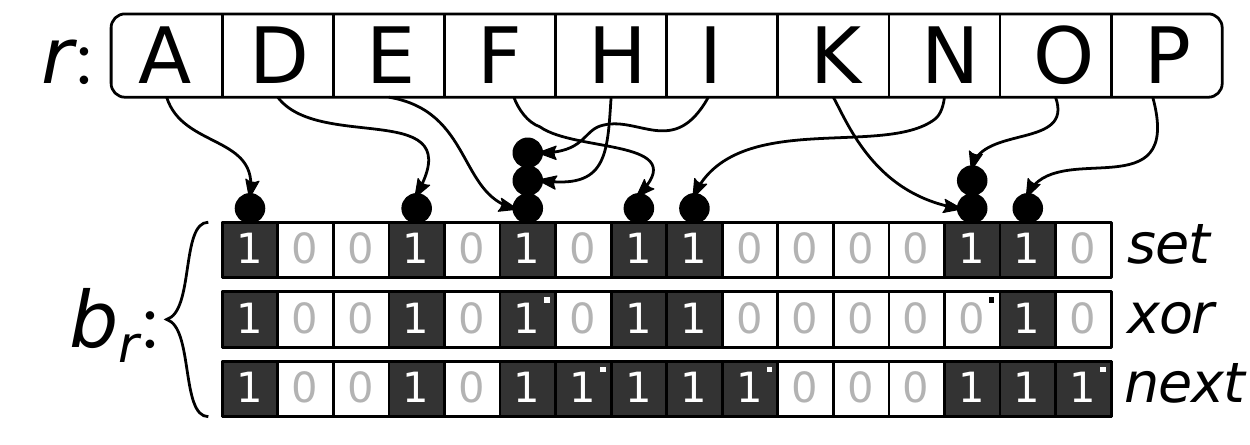}
\caption{Bitmap example\label{fig:bitmapexample}}
\end{figure}

\subsection{Overlap Upper-bound using Bitmaps}
\label{sec:bitmap:upperbound}

Given two sets $r$ and $s$ and their respective bitmaps $b_s$ and $b_r$ produced by one of the three $bitmap$ functions (Section \ref{sec:bitmap:definitions}), this section shows how to infer an upper-bound for the overlap similarity function ($|r \cap s|$). Whenever the upper-bound is not sufficient to attain the similarity threshold $\tau$ given as input, then the pair of sets $r$ and $s$ may be pruned.

Intuitively, for each occurrence of a different bit in the same position $i$ in bitmaps $b_s$ and $b_r$ (i.e. $b_s[i] \neq b_r[i]$), we can infer that there is at least one element from $r$ or $s$ that does not occur in the intersection $r \cap s$ of the sets. Given that, it can be found that the maximum overlap between these sets is the sum of the set's sizes $|r|+|s|$ minus the number of different bits (\textit{xor}) in their bitmaps (also known as \textit{hamming distance}). 

In order to prove that the overlap upper bound holds for the Bitmap-Set, Bitmap-Xor and Bitmap-Next, it must be noted that the bitmap generation functions may be represented as a series of algebraic operations. Let $\hat{h}(s_k)$ be the bitmap constructed by a single bit set at position $h(s_k)$. For instance, if $h(s_k)=4$ and $b=8$, then \linebreak $\hat{h}(s_k)=00010000$. Functions Bitmap-Set, Bitmap-Xor and Bitmap-Next are related to a series of binary operations $\ast$ over the bitmaps, as in $\hat{h}(s_1) \ast \hat{h}(s_2) \ast \cdots \ast \hat{h}(s_n)$. For simplicity, this series will also be represented by the prefix notation $\ast(s_1,s_2,\cdots,s_n)$. The $\ast$ operation is associated to the logic of each loop iteration in Algorithms \ref{algo:bitmap-set}, \ref{algo:bitmap-xor} and \ref{algo:bitmap-next}. For instance, in Bitmap-Xor, $\ast$ is the \textit{xor} operation and, in Bitmap-Set, $\ast$ is the \textit{or} operation. The Bitmap-Next $\ast$ operation is constructed with the same logic intrinsic to the lines 7 to 10 in Algorithm \ref{algo:bitmap-next}, where the collisions are chained to the next unset bit until the bitmap is saturated. It must be noted that, in the three methods, $\ast$ is commutative ($(x \ast y) = (y \ast x)$) and associative ($(x \ast y) \ast z) = x \ast (y \ast z)$), so if we modify the order of the operands $s_1, s_2, \cdots, s_n$, the resultant bitmap will be the same (e.g. $\ast(s_1, s_2, s_3, s_4)=\ast(s_2, s_4, s_1, s_3)$).  Given this observation, we have Theorem \ref{theorem:lowerbound_noncommon}.

\newtheorem{theorem}{Theorem}

\begin{theorem}
\label{theorem:lowerbound_noncommon}
Let two sets $r=\{r_1,r_2,\cdots,r_{|r|}\}$ and \linebreak $s=\{s_1,s_2,\cdots,s_{|s|}\}$ and their bitmaps $b_r$ and $b_s$ produced by Bitmap-Set, Bitmap-Xor or Bitmap-Next. 
The overlap $|r\cap s|$ is restricted to the upper bound defined by Equation \ref{eq:upper_bound_overlap}

\begin{equation}
|r\cap s| \le \left\lfloor \frac{|r|+|s|-popcount(b_r \oplus b_s)}{2} \right\rfloor
\label{eq:upper_bound_overlap}
\end{equation}

\end{theorem}

\begin{proof}
Let the common elements of sets $r$ and $s$ be $C = r \cap s$ and the non-common elements be $U = (r \setminus s) \cup (r \setminus s)$. Since the bitmap generation functions can be represented by an associative and commutative operation $\ast$, the order of their operands does not affect the final bitmap. So, without loss of generality, the elements of sets $r$ and $s$ are rearranged such that the first $|C|$ elements are the common elements, i.e., $\{r_1, r_2, \cdots, r_{|C|}\} = \{s_1, s_2, \cdots, s_{|C|}\}=C$. Since the first $|C|$ elements of sets $r$ and $s$ are the same, their partial bitmaps $b'_r=\ast(\{r_1, r_2, \cdots, r_{|C|}\})$ and $b'_s=\ast(\{s_1, _2, \cdots, s_{|C|}\})$ will also be the same (i.e. there are no different bits in $b'_r$ and $b'_s$).

Then, each remaining operand (non-common elements in $U$) is able to change at most one bit in the partial bitmaps $b'_r$ and $b'_s$. So, if there is one different bit from the final bitmaps $b_r$ and $b_s$, surely it was caused by a non-common element and, since each element can change at most one bit, then the number of different bits is lower or equal than the number of non-common elements $|U|$. The number of different bits from $b_r$ and $b_s$ can be counted by \linebreak $popcount(b_r \oplus b_s)$ (hamming distance). So, the lower bound for $|U|$ can be defined by Equation \ref{eq:lower_bound_hamming}.

\begin{equation}
|U| = |r \setminus s| + |s \setminus r| \ge popcount(b_r \oplus b_s)
\label{eq:lower_bound_hamming}
\end{equation}

Considering that $|r| + |s| = |r \setminus s| + |s \setminus r| + 2|r \cap s|$, Equation \ref{eq:lower_bound_hamming} can be directly transformed into the upper bound defined by Equation \ref{eq:upper_bound_overlap}.
\end{proof}

Figure \ref{fig:bitmapfilter} shows an example of two bitmaps $b_r$ and $b_s$ generated by the Bitmap-Set method. Then, the hamming distance is calculated using bitwise operations ($popcount(b_r \oplus b_s)$). In Figure \ref{fig:bitmapfilter}, the $b_r \oplus b_s$ operation produces a word with 4 ones and, considering Equation \ref{eq:upper_bound_overlap}, the overlap upper bound is $8$. Using the Bitmap-Xor and Bitmap-Next methods, the overlap upper bound would be $7$.

\begin{figure}[hbtp]
\centering
\includegraphics[width=5cm]{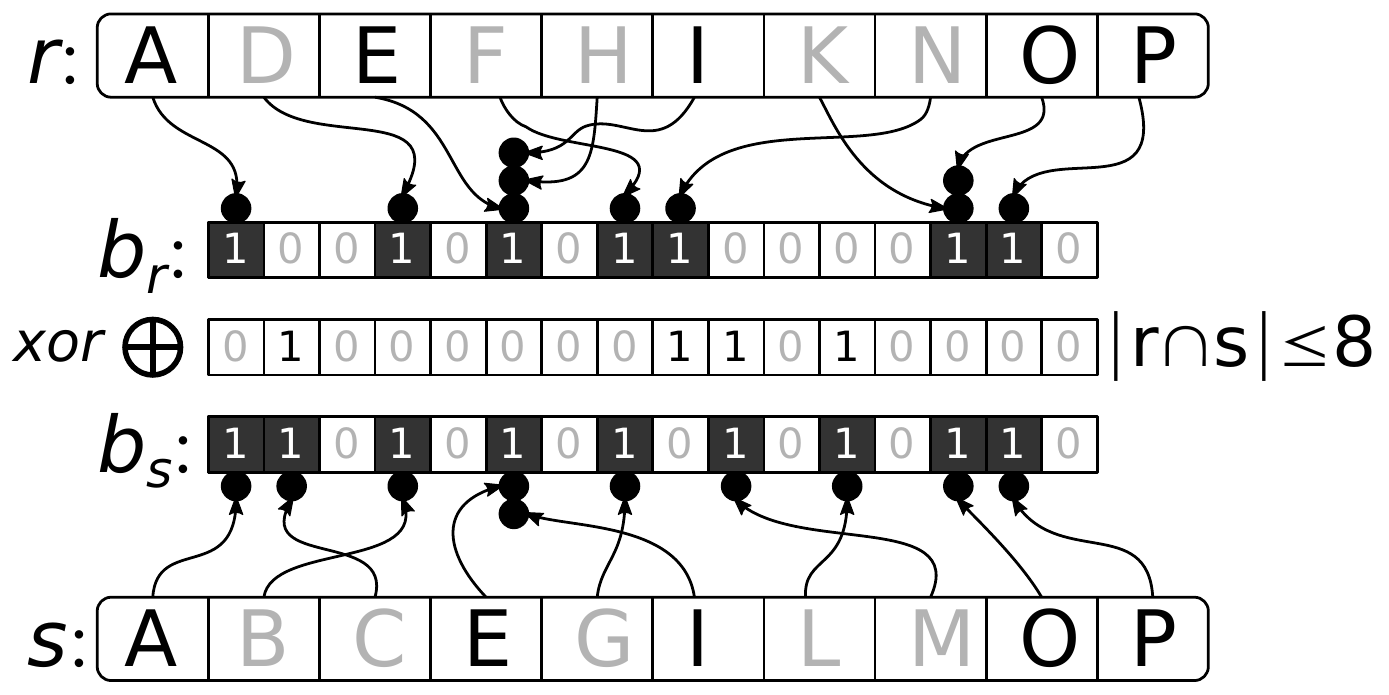}
\caption{Overlap Upper-bound\label{fig:bitmapfilter}}
\end{figure}

\subsection{Expected Overlap Upper-bound}
\label{sec:bitmap:expected_upperbound}

Due to collisions in the hash functions $h(t)$, two sets may produce similar bitmaps even if they do not share similar tokens. As the probability of collision is increased, the filtering effectiveness is reduced because it will loose the capability of distinguishing between similar and dissimilar sets. Intuitively, the chance for a collision is increased whenever the number of tokens in the original sets $r$ and $s$ is increased or the size $b$ of the bitmaps is reduced. This produces a trade-off between the filtering efficiency (precision) and the reduction of the bitmap dimension.

Assuming that the hash function distributes the tokens uniformly in each $b$ bits of the bitmaps, a probability analysis is able to determine, for each bitmap generation method, the expected number of collisions given two sets with $n$ random tokens. Given the expected number of collisions, it is also possible to infer the expected overlap upper-bound using Equation \ref{eq:upper_bound_overlap}. Equations \ref{eq:e_mark}, \ref{eq:e_xor}, and \ref{eq:e_next} present the expected overlap upper bound $\mathbb{E}(b,n)$ for a given $b$ bitmap size and $n$ tokens, considering the Bitmap-Set, Bitmap-Xor and Bitmap-Next respectively. These equations are based on the estimation of collision probabilities in hash functions \cite{PEYRAVIAN1998171}. Furthermore, the probabilities were verified by a Monte-Carlo simulation, where 100,000 random pairs of sets $r$ and $s$ were generated with $n$ tokens and the overlap upper-bound was obtained for each pair. Running the simulation for each $n \in [1,128]$ and for $b=64$, the average results produced by the simulation were almost identical to those derived from Equations \ref{eq:e_mark}, \ref{eq:e_xor}, and \ref{eq:e_next}, with an average error below 0.012\%. 

\begin{align} 
{\underset{\mathclap{mark}}{\mathlarger{\mathbb{E}}}(b,n)} &= n+\frac{(b-1)^{2n}}{b^{2n-1}}-\frac{(b-1)^{n}}{b^{n-1}}
\label{eq:e_mark}
\\
{\underset{\mathclap{xor}}{\mathlarger{\mathbb{E}}}(b,n)} &= {n-\frac{b}{2}{\mathlarger{\sum}_{\mathclap{k=1,3,\cdots}}^{2n}}} {{2n}\choose{k}}(\frac{1}{b})^k(\frac{b-1}{b})^{2n-k} 
\label{eq:e_xor}
\\
{\underset{\mathclap{next}}{\mathlarger{\mathbb{E}}}(b,n)} &= min(\frac{n^2}{b},n)
\label{eq:e_next}
\end{align}

Figure \ref{fig:expected_upper_bound} plots the upper bounds defined by Equations \ref{eq:e_mark}, \ref{eq:e_xor}, and \ref{eq:e_next} considering $b=64$ and $n$ varying from 1 to 256. The left $y$-axis presents a scale with the normalized overlap, given by the overlap divided by $n$, and the right $y$-axis presents a scale with the equivalent Jaccard upper bound (Table \ref{tb:simfuncs}). As the $y$-value increases, more difficult will be to distinguish between similar and non-similar pairs (the lower the $y$ value, better is the method). A $y$-value of $1$ represents an upper bound equal to the number of tokens $n$ (useless filter), and a $y$-value of $0$ represents a zero upper-bound. For example, considering the Bitmap-Set and the Bitmap-Xor, the expected normalized overlap upper bound at $n=55$ is around $0.72$ (or $0.84$ considering the equivalent Jaccard metric). So, a Jaccard threshold of $\tau_j=0.84$ will be expected to filter 50\% of dissimilar pairs of sets composed of $55$ random tokens. 

\begin{figure}[t]
\includegraphics[width=8cm]{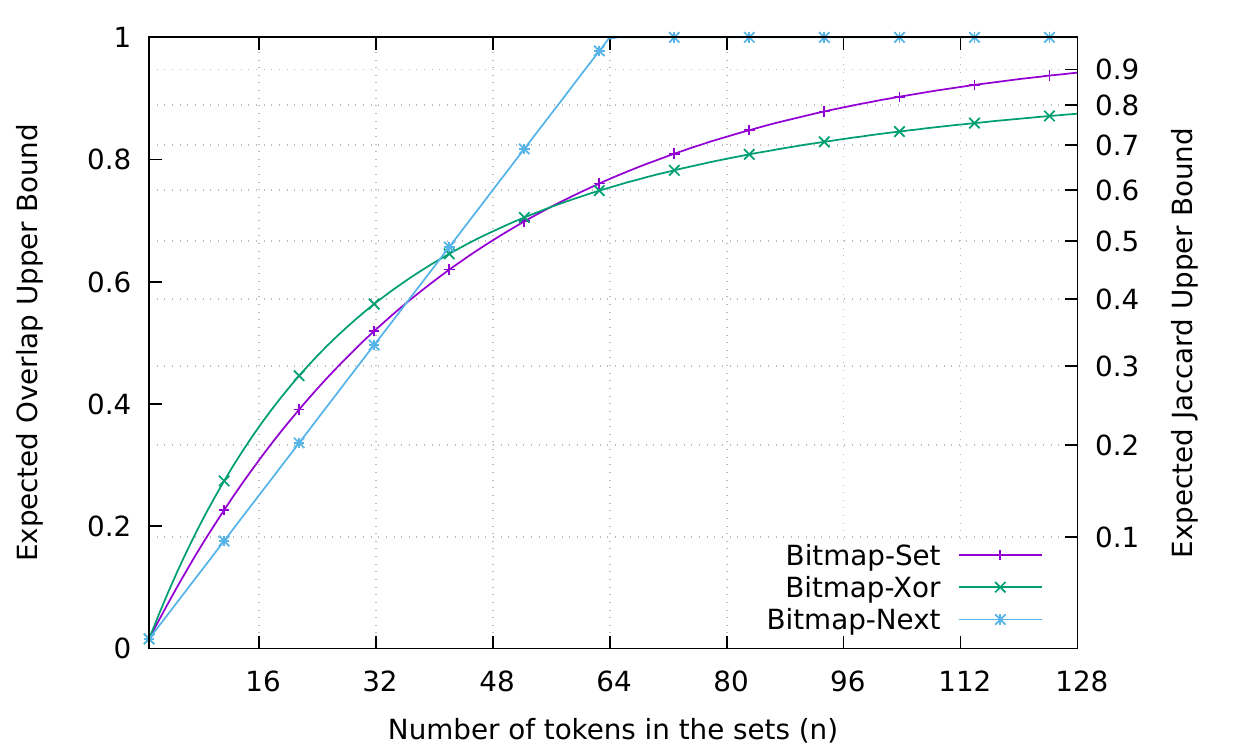}
\caption{Expected upper bounds ($b=64$)\label{fig:expected_upper_bound}}
\end{figure}

\subsection{Cutoff Point for Bitmap Filter}
\label{sec:bitmap:cutoff}

 Given a bitmap size $b$ and the overlap threshold $\tau$, the cutoff point $\omega(b,\tau)$ is defined as the maximum number of tokens $n$ at which the Bitmap Filter can efficiently distinguish between similar and dissimilar sets. The cutoff point $w(b,\tau)$ can be defined in terms of the $\mathbb{E}(b,n)$ function (Equations \ref{eq:e_mark}, \ref{eq:e_xor} and \ref{eq:e_next}) such that $w(b,\tau)=n$ implies $\mathbb{E}(b,n)=\tau$. In Figure \ref{fig:expected_upper_bound}, the average normalized overlap upper bound at $b=64$ and $n=55$ is around $0.72$. So, using a threshold $\tau=0.72$, the Bitmap Filter precision will be effective with sets with up to $n=55$ tokens and, with more tokens, the precision will drop-off significantly, since the filter will not be able to discard the majority of dissimilar pairs. So, the cutoff point for $\tau_j=0.72$ will be $\omega(64,0.72)=55$.

\begin{figure}[t]
\includegraphics[width=8cm]{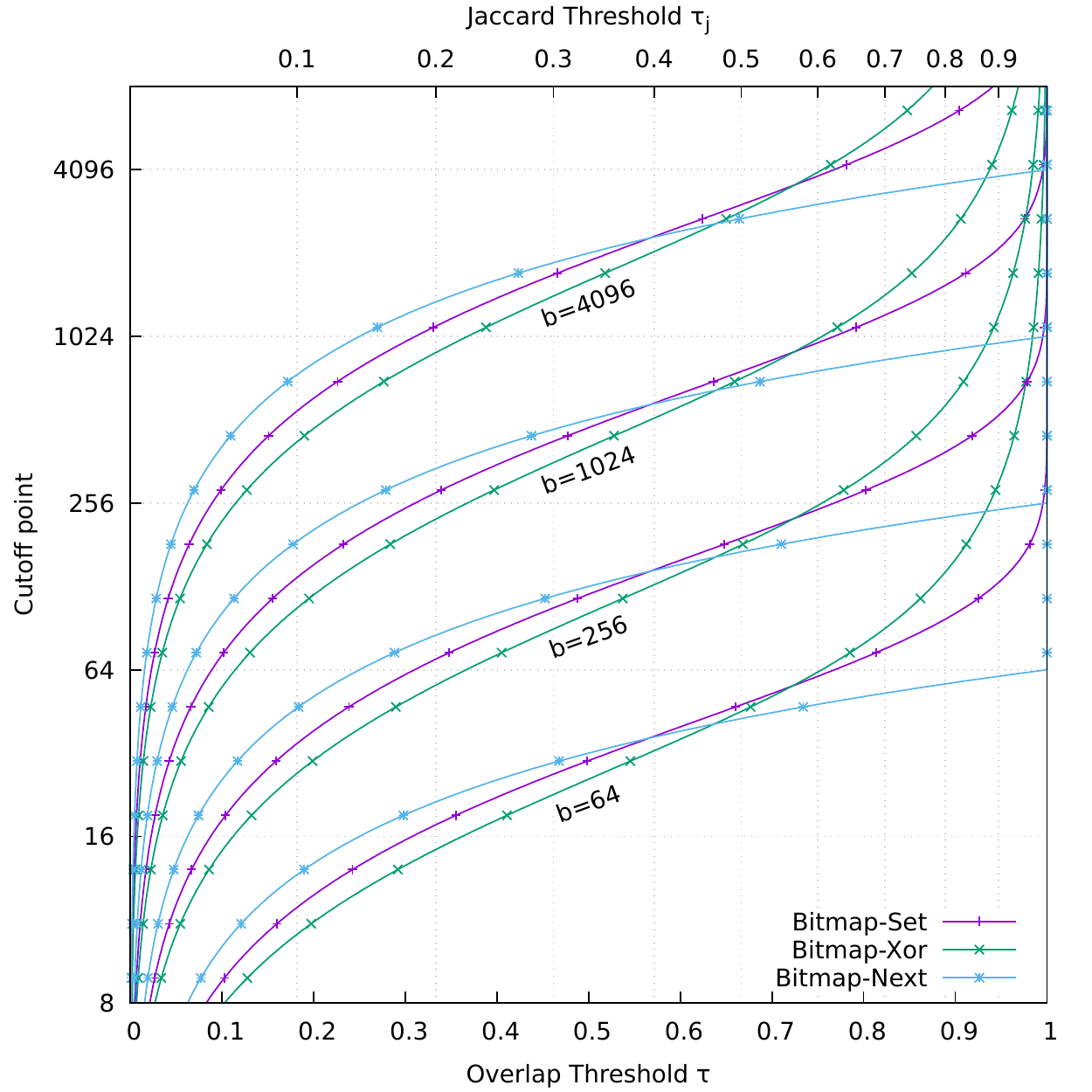}
\caption{Cutoff values (higher is better)\label{fig:cutoff}}
\end{figure}

Figure~\ref{fig:cutoff} presents the cutoff point for Bitmap-Set, Bitmap-Xor and Bitmap-Next, with sizes $b \in \{64, 256, 1024, 4096\}$. A similar plot pattern is observed for any $b$ greater than 64. Analyzing the cutoff difference between Bitmap-Xor and Bitmap-Set at high thresholds (e.g. above $0.8$), it is noticeable that Bitmap-Xor presents higher cutoffs. For instance, with $b=1024$ and Jaccard threshold $\tau_j=0.9$, Bitmap-Xor attained a cutoff at $n=4983$, whereas the cuttoff point of Bitmap-Set is $2129$. In other words, the Bitmap-Xor will still be effective with $2.3\times$ more tokens than Bitmap-Set. For Jaccard threshold $0.8$, this proportion is $1.47\times$.

In Figures~\ref{fig:expected_upper_bound} and~\ref{fig:cutoff}, it can be seen that there is a preferable bitmap generation method for each threshold range: Bitmap-Next presents the higher cutoff when $\tau \le 0.56$; Bitmap-Set is slightly better between $0.56 < \tau < 0.73$; Bitmap-Xor is preferred when $\tau \ge 0.73$. This pattern is observed for any value of $b$ greater than 64. So, a combined bitmap generation can be created using the preferred bitmap generation method for each $\tau$ interval (Algorithm \ref{algo:bitmap-combined}). Figure~\ref{fig:cutoff_combined} presents the cutoff point $\omega(b,\tau)$ for the combined bitmap generation method, where the $y$-axis scale is relative to bitmap size $b$.

\begin{algorithm}
\caption{Bitmap-Combined \label{algo:bitmap-combined}}
\begin{algorithmic}[1]
\Function{Bitmap-Combined}{$s = \{s_1,s_2,\cdots,s_n\}$}
\If{$\tau \le 0.56$} \State \Return \Call{Bitmap-Next}{s} 
\ElsIf{$\tau \ge 0.73$} \State \Return \Call{Bitmap-Xor}{s} 
\Else{ }\State \Return \Call{Bitmap-Set}{s}
\EndIf
\EndFunction
\end{algorithmic}
\end{algorithm}

\begin{figure}[t]
\includegraphics[width=8cm]{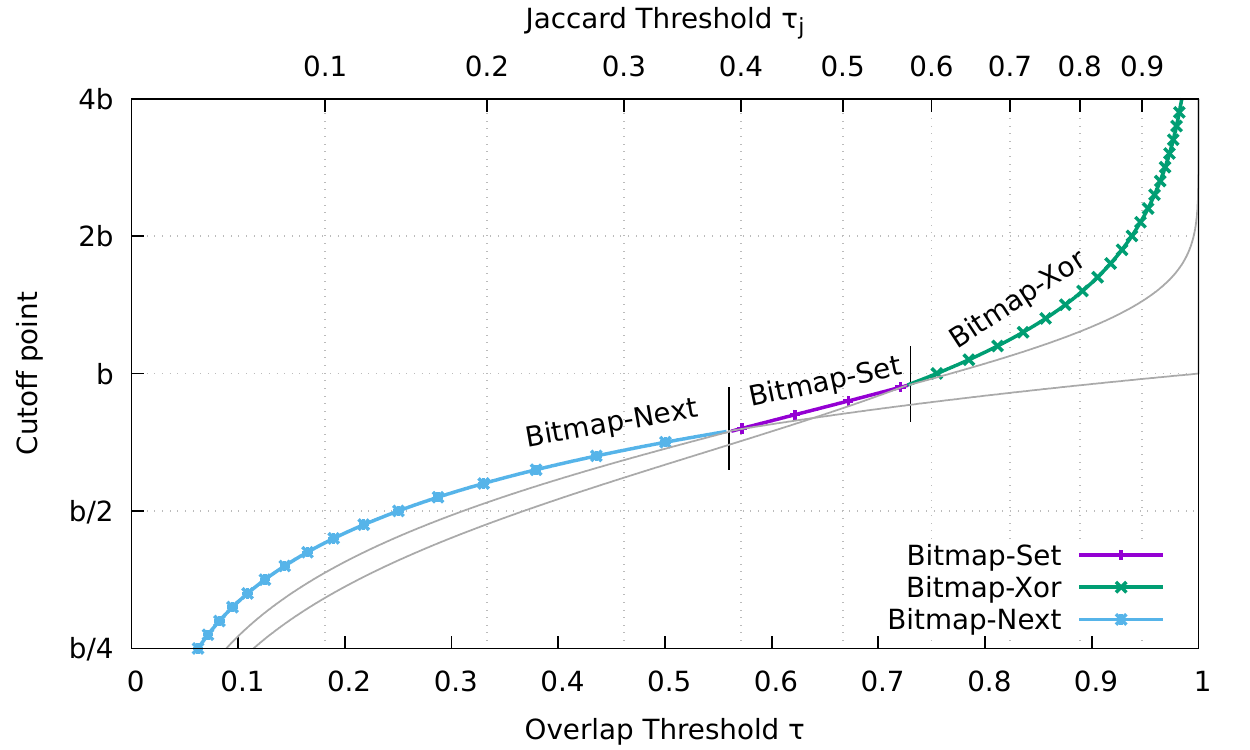}
\caption{Best cutoffs\label{fig:cutoff_combined}}
\end{figure}

\subsection{Bitmap Filter Pseudocode}
\label{sec:bitmap:algorithm}

The pseudocode for the Bitmap Filter is presented in Algorithm \ref{algo:proposal}. This pseudocode is suitable for the Filter-Verification Framework (Algorithm \ref{algo:basic}) and it is divided in two sections: the initialization procedure (lines 1-4) is run only once for the entire execution; the Bitmap\_Filter function (lines 5-10) can be executed at \emph{filter}$_2$ or \emph{filter}$_3$ in Algorithm \ref{algo:basic}.

The initialization precomputes the bitmap for every set in the collections $R$ and $S$. As stated in Algorithm \ref{algo:bitmap-combined}, the Bitmap generation method is selected accordingly to threshold $\tau$. The initialization also calculates the cutoff point (Section \ref{sec:bitmap:cutoff}) for the selected bitmap size $b$ and threshold $\tau$. The cutoff calculation is precomputed using the combined method presented in Figure \ref{fig:cutoff_combined}.

Therefore, for a given candidate pair $r$ and $s$, the function Bitmap\_Filter executes the proposed filter. If the length $|r|$ is above the cutoff point (line 7), the Bitmap Filter will be ignored. Otherwise, the filter calculates the overlap upper bound (line 8) as defined in Equation \ref{eq:upper_bound_overlap}. If the upper bound is below threshold $\tau$ (line 9), then the filter prunes the pair.

\begin{algorithm}
\caption{Proposed Algorithm - Bitmap Filter \label{algo:proposal}}
\vspace{0.1cm}
\begin{algorithmic}[1]
\Procedure{initialization}{$R$,$S$,$\tau$,$b$}
\For{$r \in \{R\cup S\}$}
\State $b_r \leftarrow \Call{Bitmap}{r}$
\EndFor
\State $cuto\mathit{ff}\_point \leftarrow \omega(b,\tau)$
\EndProcedure
\Function{bitmap\_filter}{r,s,$\tau$}
	\State $skip \leftarrow false$
	\If {$|r| \le cuto\mathit{ff}\_point$} 
		\State $upper\_bound = \dfrac{|r|+|s|-popcount(b_r \oplus b_s)}{2}$
		\If {$upper\_bound < \tau$} $skip = true$; \EndIf
	\EndIf
	\State \Return $skip$
\EndFunction
\end{algorithmic}
\vspace{0.1cm}
\end{algorithm}

\section{Implementation Details}
\label{sec:implementations}

The Bitmap Filter will be evaluated using CPU and GPU implementations, which are detailed in this section.

\subsection{CPU Implementation}
\label{sec:implementations:cpu}

In order to assess the Bitmap Filter in CPU, we added it to the state-of-the-art implementation of \cite{mann2016} using the available source code. Four algorithms were modified: AllPairs, PPJoin, GroupJoin, and AdaptJoin. These algorithms follow the Filter-Verification Framework (Algorithm \ref{algo:basic}) and, considering the peculiarities of each algorithm, the Bitmap Filter was introduced as \emph{filter}$_2$ (line \ref{algo:basic:filter_2}) or \emph{filter}$_3$ (line \ref{algo:basic:filter_3}). If the filter is applied in the candidate generation loop (\emph{filter}$_2$), the bitmap test may be applied more than once for the same candidate pair. If the filter is applied in the verification stage loop (\emph{filter}$_3$), the bitmap test will be applied only once for each unique candidate pair. 

For AllPairs, PPJoin, and GroupJoin we chose to insert the Bitmap Filter in the verification loop (\emph{filter}$_3$). In the specific case of the GroupJoin algorithm, the filter is applied after the grouped candidate pairs are expanded, in such way that the filter is applied to individual elements. 

For the AdaptJoin, we verified that the candidate list is relatively small at the verification loop, so the Bitmap Filter is applied at the candidate generation (\emph{filter}$_2$). The filter is executed only during the first $\ell$-prefix filter iterations (i.e. in the $1$-prefix schema computation).

The bitmaps were implemented with multiple of 64-bit words and the population count operation was done using the \texttt{\_\_builtin\_popcountll} gcc intrinsic function. Using proper compilation flags, the gcc converts this function to the POPCNT hardware instruction for population count~\cite{nehalem2008inside}. 

\subsection{GPU Implementation}
\label{sec:implementations:gpu}

The Bitmap Filter relies on bitwise instructions (such as \textit{xor} and population count) that are efficiently implemented in Graphical Processor Units (GPUs). Compute Unified Device Architecture (CUDA)\cite{CUDAGuide8.0} is the Nvidia general purpose framework that allows parallel execution of procedures (kernels) in Nvidia GPUs. In CUDA, each kernel is executed in parallel by several threads, which are grouped into independent blocks. Threads in the same block share data using a fast shared memory. The CUDA framework also allows the data transfer between host and GPU devices, in order to supply inputs to the kernel and send the output back to the host machine.

In order to show the potential of the Bitmap Filter in the CUDA architecture, we implemented the pseudocode of Algorithm \ref{algo:gpu} in GPU, where $R$ is the self-joined collection, $\tau_j$ is the Jaccard threshold and $cand$ is the return list containing a list of candidate pairs. This code is equivalent to the Na\"ive Algorithm (Algorithm \ref{algo:naive}) with the addition of the Length Filter (Section \ref{sec:filter:length}) and Bitmap Filter (Section \ref{sec:bitmap_filter}). The Bitmap-Xor generation method was used and the cutoff point was disabled. The kernel call (line \ref{algo:gpu:kernelcall}) executes $B$ blocks with $T$ threads each, and each thread receives a unique thread id $i$ (line \ref{algo:gpu:thread_id}). 

The kernel may be invoked many times until all sets are processed (line \ref{algo:gpu:max_i}). Each thread compares set $R[i]$ with all the previous sets $R[j_0\le j<i]$, where $j_0$ is the first index containing a set not filterable by the Length Filter (line \ref{algo:gpu:lengthfilter}). If the Bitmap Filter finds a possible candidate pair, this pair is included in a thread-local candidate list (lines \ref{algo:gpu:bitmap_filter}-\ref{algo:gpu:local_cand}). Each thread-local candidate list can hold up to 2048 candidates. If the thread-local list becomes saturated, all the remaining sets will be considered candidates and will be verified by the CPU.

At the end of the kernel execution (line \ref{algo:gpu:global_cand}), the thread-local candidate lists are concatenated in a global candidate list, without empty spaces. This operation, called stream compaction or reduction, reduces the amount of data transferred to the host. Finally, in the host code, all candidate pairs are verified and the similar pairs are included in the final result list (lines \ref{algo:gpu:verify0}-\ref{algo:gpu:return}).

\begin{algorithm}
\caption{GPU Algorithm \label{algo:gpu}}
\vspace{0.1cm}
\begin{algorithmic}[1]
\Procedure{GPUKernel}{$R$, $\tau_j$, $cand$}
\State $i \leftarrow thread\_id$ \label{algo:gpu:thread_id}
\If{$i > |R|$} return \EndIf \label{algo:gpu:max_i}
\State $j_0 \leftarrow \argmin_x \{|R[x]|\ge \lceil \tau_j\cdot|R[i]| \rceil\}$ \Comment{length filter} \label{algo:gpu:lengthfilter}
\For{$j \in [j_0..i)$}
\If {$bitmap\_filter(R[i],R[j],\tau_j)$} \label{algo:gpu:bitmap_filter}
\State $local\_cand \leftarrow local\_cand \cup (r,s)$ \label{algo:gpu:local_cand}
\EndIf
\EndFor
\State $cand \leftarrow block\_compact(local\_cand)$ \label{algo:gpu:global_cand}
\EndProcedure

\Procedure{HostCode}{$R$,$\tau_j$}
\State GPUKernel$\lll B,T\ggg(R,\tau_j, cand)$ \label{algo:gpu:kernelcall}
\For{$(r,s) \in cand$} \label{algo:gpu:verify0}
\If {$verify(r,s,\tau_j)$}
\State $pairs \leftarrow pairs \cup (r,s)$
\EndIf
\EndFor \label{algo:gpu:verify1}
\State \Return $pairs$ \label{algo:gpu:return}
\EndProcedure
\end{algorithmic}
\vspace{0.1cm}
\end{algorithm}

The GPU algorithm may produce up to  $|R|^2$ Bitmap Filter computations and its implementation was not intended for collections with more than a million elements. Nevertheless, it was implemented in order to show that even a simple algorithm can take advantage of the Bitmap Filter and become very competitive when compared to the state-of-the-art algorithms. So, we claim that more sophisticated GPU implementations may benefit more from Bitmap Filter, opening the opportunity of new improvements for the Set Similarity Join problem.

\section{Experimental Results}
\label{sec:experimental_results}

In order to create a baseline, we replicated the experiments conducted by \cite{mann2016}. Table \ref{tb:collections} presents the characteristics of the chosen collections. Comparing this table with \cite{mann2016}, slight variations can be seen in \myinput{dblp} collection, probably due to a different charset. Also, \myinput{flickr}, \myinput{netflix} and \myinput{spot} collections were no more available for download in the links provided by \cite{mann2016}. We generated new \myinput{uniform} and \myinput{zipf} collections with the same methodology \cite{mann2016}, using a Poisson distribution for the set sizes. All collections were preprocessed and the sets in the collections were sorted by size and, in case of a tie, they were sorted lexicographically by the token ids. We observed that the lexicographical ordering speeds up all the algorithms, as stated by \cite{mann2016}.

\begin{table}[t]
\centering
\caption{Collections used in experiments}
\setlength\tabcolsep{5.0pt}
\begin{tabular}{|c|r|r|r|r|r|}
\hline 
Collection& \multicolumn{1}{c|}{\# of} & \multicolumn{3}{c|}{set size} & \multicolumn{1}{c|}{\# uniq.} \\
\cline{3-5}
name & \multicolumn{1}{c|}{sets} & \multicolumn{1}{c|}{max} & \multicolumn{1}{c|}{avg} & \multicolumn{1}{c|}{med} &  \multicolumn{1}{c|}{tokens} \\ 
\hline 
\myinput{aol}& 	 	10154743& 	 245& 	  3.01&		3&	3873246	\\
\myinput{bms-pos}& 	320285& 	 164& 	  9.30&		7&	1657	\\
\myinput{dblp}& 		100000& 	 717& 	  106.28&	103&	3801	\\
\myinput{enron}& 	 	245615& 	 3162& 	  135.19&	86&	1113220	\\
\myinput{kosarak}& 	606770& 	 2498& 	  11.93&	5&	41275	\\
\myinput{livej}& 	 	3061271& 	 300& 	  36.44&	17&	7489073	\\
\myinput{orkut}& 	 	2732271& 	 40425&   119.67& 	29&	8730857	\\
\myinput{uniform}& 	100000& 	 25&	  9.99&	10&	220		\\
\myinput{zipf}& 	 	100000& 	 86& 	  49.99& 	50&	101584	\\
\hline 
\end{tabular} 
\label{tb:collections}
\end{table}

Figure \ref{fig:setsizes} shows the set size distribution in the collections. \myinput{Orkut}, \myinput{enron} and \myinput{kosarak} present a very large tail to their right, whereas \myinput{dblp}, \myinput{uniform} and \myinput{zipf} collections present a more symmetrical frequency distribution. Figure \ref{fig:tokcount} presents the number of occurrences of the tokens in some collections, where a Zipf distribution is often observed. This kind of distribution increases the efficiency of the Prefix Filter \cite{mann2016}.

\begin{figure}[t]
\includegraphics[width=8cm]{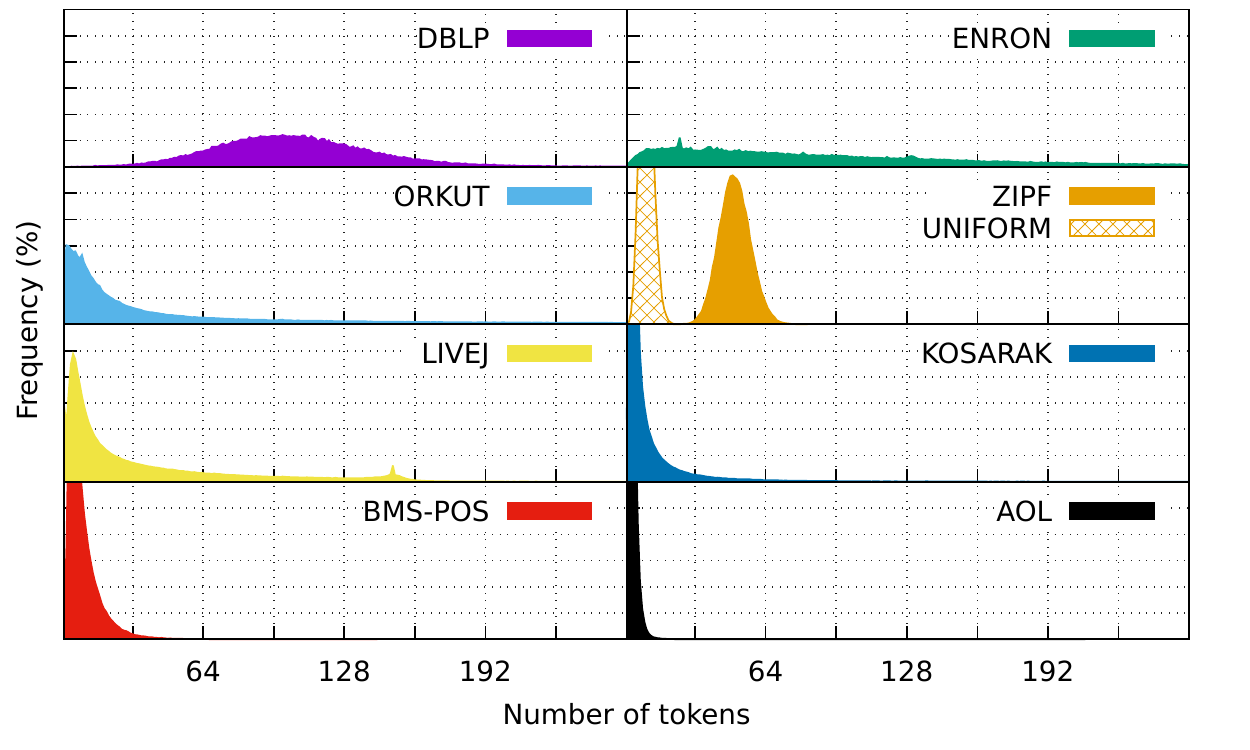}
\caption{Set size frequency\label{fig:setsizes}}
\end{figure}

\begin{figure}[t]
\includegraphics[width=8cm]{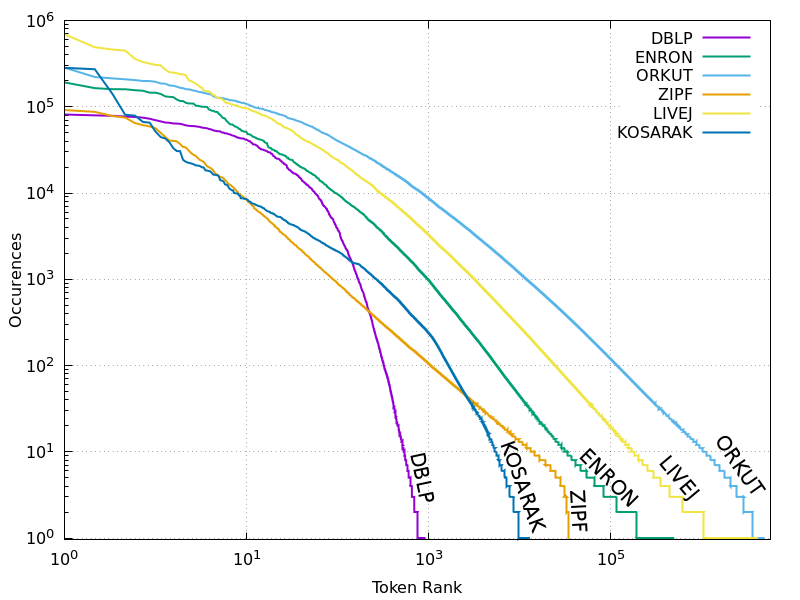}
\caption{Token occurrence histogram\label{fig:tokcount}}
\end{figure}

The experiments were conducted in a dedicated machine with an Intel i7-3770 CPU  (4 cores) at 3.40GHz with 8 GB RAM and an Nvidia GeForce GTX 980 Ti (2048 CUDA cores) with 6 GB RAM, running a CentOS Linux 7.2. The \texttt{gcc} compiler was used with the flags \texttt{-O3} and \texttt{-march=native}. As in \cite{mann2016}, all the presented runtimes were taken as average of 5 runs.

\subsection{CPU Experiments}

The four best state-of-the-art algorithms reported by \cite{mann2016} were selected for the experiments: AllPairs (\myalgo{ALL}), AdaptJoin (\myalgo{ADA}), GroupJoin (\myalgo{GRO}) and PPJoin (\myalgo{PPJ}). For the baseline experiments, we used the original source code provided by \cite{mann2016}. For the Bitmap Filter experiments, we used the modified source code proposed in Section \ref{sec:implementations}, with the Bitmap-Combined generation method (Algorithm \ref{algo:bitmap-combined}). The default bitmap size was $b=64$, but for the collections with large median set sizes (\myinput{dblp}, \myinput{enron} and \myinput{zipf}), the bitmap size was increased to $b=128$. The selected hash function was $h(t)=t \bmod b$.
In order to distinguish the modified algorithms that applied the Bitmap Filter, they will be referenced as: AllPairs-BF (\myalgo{ALL-BF}), AdaptJoin-BF (\myalgo{ADA-BF}), GroupJoin-BF (\myalgo{GRO-BF}) and PPJoin-BF (\myalgo{PPJ-BF}).

The experiments were conducted over 9 collections (self-joined) and 8 threshold $\tau$ varying from $0.50$ to $0.95$, resulting in 72 different input combinations. Each input was assessed by the two groups of algorithms: original algorithms (baseline) and the modified algorithms (with Bitmap Filter). In total, there are 288 experiments for each group of algorithms.

\subsubsection{CPU Execution times}

 Table \ref{tb:results_baseline} shows the running times of the original algorithms in columns ``orig.'' (baseline) and the running times of the modified algorithms in columns ``+BF'' (with Bitmap Filter). The times do not include the preprocessing overhead (loading files and sorting the input).


\begin{table*}[t]
\centering
\caption{CPU runtimes (in seconds) comparing the baseline execution (orig.) with the Bitmap Filter (BF) execution. Speedups are highlighted in \textbf{bold} and slowdowns in \textbf{\textcolor{gray}{gray}}. Inside each group with four algorithm executions, the best runtime is marked with asterisk (*).}
\label{tb:results_baseline}
\scriptsize
\setlength\tabcolsep{3.0pt}
\begin{tabular}{|@{\hspace{0.1cm}}c@{ }|p{0.5cm}||r@{ }|r@{ }||r@{ }|r@{ }||r@{ }|r@{ }||r@{ }|r@{ }||r@{ }|r@{ }||r@{ }|r@{ }||r@{ }|r@{ }||r@{ }|r@{ }|}
\cline{3-18}
\multicolumn{2}{c|}{}& \multicolumn{16}{c|}{Threshold $\tau_j$}\\
\cline{3-18}
\multicolumn{2}{c|}{} & \multicolumn{2}{c||}{0.5} & \multicolumn{2}{c||}{0.6} & \multicolumn{2}{c||}{0.7} & \multicolumn{2}{c||}{0.75} & \multicolumn{2}{c||}{0.8} & \multicolumn{2}{c||}{0.85} & \multicolumn{2}{c||}{0.9} & \multicolumn{2}{c|}{0.95} \arraybackslash \\
\cline{3-18}
\multicolumn{2}{c|}{} & \multicolumn{1}{c|}{orig.} & \multicolumn{1}{c||}{+BF} & \multicolumn{1}{c|}{orig.} & \multicolumn{1}{c||}{+BF} & \multicolumn{1}{c|}{orig.} & \multicolumn{1}{c||}{+BF} & \multicolumn{1}{c|}{orig.} & \multicolumn{1}{c||}{+BF} &\multicolumn{1}{c|}{orig.} & \multicolumn{1}{c||}{+BF} & \multicolumn{1}{c|}{orig.} & \multicolumn{1}{c||}{+BF} & \multicolumn{1}{c|}{orig.} & \multicolumn{1}{c||}{+BF} & \multicolumn{1}{c|}{orig.} & \multicolumn{1}{c|}{+BF} \\
\hline

\multirow{4}{*}{\rotatebox[origin=c]{90}{\myinput{aol}}}
&\multicolumn{1}{c||}{ADA}&{437.3}&*\textbf{113.1}&{119.3}&*\textbf{31.6}&{24.15}&\textbf{7.74}&{20.27}&\textbf{6.50}&{12.22}&\textbf{4.80}&{6.40}&\textbf{3.80}&{5.84}&\textbf{3.63}&{5.81}&\textbf{3.60}\\
&\multicolumn{1}{c||}{ALL}&{296.1}&\textbf{183.4}&{75.9}&\textbf{46.1}&{9.81}&*\textbf{6.49}&{6.30}&*\textbf{4.60}&{3.16}&*\textbf{2.40}&{1.54}&\textbf{1.28}&{1.34}&\textbf{1.16}&{1.33}&\textbf{1.15}\\
&\multicolumn{1}{c||}{GRO}&{219.2}&\textbf{207.3}&*{54.7}&\textbf{51.1}&*{8.89}&\textbf{7.90}&*{6.16}&\textbf{5.40}&*{2.82}&\textbf{2.41}&*{1.10}&*\textbf{0.93}&*{0.86}&*\textbf{0.74}&*{0.84}&*\textbf{0.72}\\
&\multicolumn{1}{c||}{PPJ}&*{212.2}&\textbf{166.7}&{57.1}&\textbf{44.6}&{10.72}&\textbf{8.20}&{8.07}&\textbf{6.19}&{4.27}&\textbf{3.28}&{2.19}&\textbf{1.71}&{1.89}&\textbf{1.54}&{1.87}&\textbf{1.53}\\
\hline
\multirow{4}{*}{\rotatebox[origin=c]{90}{\myinput{bms-pos}}}
&\multicolumn{1}{c||}{ADA}&*{24.49}&*\textbf{8.79}&*{8.47}&*\textbf{2.96}&*{2.61}&*\textbf{0.98}&{1.63}&*\textbf{0.61}&{0.88}&*\textbf{0.34}&{0.42}&*\textbf{0.18}&{0.26}&\textbf{0.12}&{0.20}&\textbf{0.10}\\
&\multicolumn{1}{c||}{ALL}&{36.63}&\textbf{15.72}&{10.63}&\textbf{5.75}&{3.04}&\textbf{1.75}&{1.69}&\textbf{1.03}&{0.80}&\textbf{0.50}&{0.29}&\textbf{0.20}&{0.13}&\textbf{0.09}&{0.08}&\textbf{0.06}\\
&\multicolumn{1}{c||}{GRO}&{36.11}&\textbf{33.04}&{10.39}&\textbf{9.32}&{2.88}&\textbf{2.41}&*{1.62}&\textbf{1.32}&*{0.77}&\textbf{0.61}&*{0.28}&\textbf{0.22}&*{0.10}&*\textbf{0.08}&*{0.04}&*\textbf{0.03}\\
&\multicolumn{1}{c||}{PPJ}&{31.99}&\textbf{19.86}&{9.97}&\textbf{6.70}&{3.06}&\textbf{2.15}&{1.81}&\textbf{1.29}&{0.93}&\textbf{0.66}&{0.38}&\textbf{0.28}&{0.18}&\textbf{0.13}&{0.12}&\textbf{0.09}\\
\hline
\multirow{4}{*}{\rotatebox[origin=c]{90}{\myinput{dblp}}}
&\multicolumn{1}{c||}{ADA}&*{161.6}&*\textbf{\textcolor{gray}{172.7}}&*{75.3}&*\textbf{64.6}&*{31.01}&*\textbf{14.20}&*{17.47}&*\textbf{7.06}&*{8.21}&*\textbf{3.37}&*{3.23}&*\textbf{1.37}&*{0.96}&*\textbf{0.42}&{0.23}&\textbf{0.10}\\
&\multicolumn{1}{c||}{ALL}&{177.1}&\textbf{\textcolor{gray}{190.7}}&{83.4}&\textbf{74.4}&{37.45}&\textbf{27.00}&{23.17}&\textbf{15.34}&{11.64}&\textbf{7.35}&{4.60}&\textbf{2.73}&{1.24}&\textbf{0.71}&*{0.16}&*\textbf{0.10}\\
&\multicolumn{1}{c||}{GRO}&{216.0}&\textbf{\textcolor{gray}{228.4}}&{107.7}&\textbf{103.0}&{39.72}&\textbf{36.72}&{22.10}&\textbf{20.24}&{10.42}&\textbf{9.48}&{4.01}&\textbf{3.74}&{1.12}&\textbf{1.05}&{0.17}&\textbf{0.15}\\
&\multicolumn{1}{c||}{PPJ}&{199.3}&\textbf{\textcolor{gray}{206.7}}&{102.9}&\textbf{93.2}&{38.86}&\textbf{31.50}&{21.72}&\textbf{16.94}&{10.21}&\textbf{7.91}&{4.00}&\textbf{3.00}&{1.11}&\textbf{0.82}&{0.17}&\textbf{0.12}\\
\hline
\multirow{4}{*}{\rotatebox[origin=c]{90}{\myinput{enron}}}
&\multicolumn{1}{c||}{ADA}&*{30.35}&*\textbf{\textcolor{gray}{31.47}}&{11.05}&\textbf{10.41}&{4.46}&\textbf{3.66}&{2.93}&\textbf{2.21}&{1.91}&\textbf{1.36}&{1.18}&\textbf{0.86}&{0.63}&\textbf{0.49}&{0.32}&\textbf{0.30}\\
&\multicolumn{1}{c||}{ALL}&{43.51}&\textbf{\textcolor{gray}{45.27}}&{13.02}&\textbf{12.19}&{3.69}&\textbf{3.05}&{1.95}&\textbf{1.56}&{1.06}&*\textbf{0.88}&{0.59}&*\textbf{0.51}&*{0.26}&*\textbf{0.24}&*{0.12}&*\textbf{0.12}\\
&\multicolumn{1}{c||}{GRO}&{32.96}&\textbf{\textcolor{gray}{33.37}}&*{9.90}&\textbf{9.70}&*{2.93}&\textbf{2.79}&*{1.69}&\textbf{1.58}&*{0.98}&\textbf{0.93}&*{0.58}&\textbf{0.55}&{0.27}&\textbf{0.25}&{0.13}&\textbf{0.13}\\
&\multicolumn{1}{c||}{PPJ}&{33.25}&\textbf{32.72}&{10.35}&*\textbf{9.41}&{3.07}&*\textbf{2.72}&{1.77}&*\textbf{1.54}&{1.03}&\textbf{0.90}&{0.59}&\textbf{0.53}&{0.27}&\textbf{0.25}&{0.12}&\textbf{0.12}\\
\hline
\multirow{4}{*}{\rotatebox[origin=c]{90}{\myinput{kosarak}}}
&\multicolumn{1}{c||}{ADA}&{93.83}&*\textbf{27.53}&{13.59}&*\textbf{4.12}&{1.53}&\textbf{0.78}&{1.03}&\textbf{0.56}&{0.71}&\textbf{0.41}&{0.47}&\textbf{0.31}&{0.36}&\textbf{0.25}&{0.28}&\textbf{0.22}\\
&\multicolumn{1}{c||}{ALL}&{58.50}&\textbf{35.59}&{8.77}&\textbf{5.43}&{1.04}&*\textbf{0.67}&*{0.59}&*\textbf{0.41}&*{0.32}&*\textbf{0.23}&*{0.15}&*\textbf{0.12}&{0.09}&*\textbf{0.08}&{0.06}&\textbf{0.06}\\
&\multicolumn{1}{c||}{GRO}&*{37.54}&\textbf{35.81}&*{5.82}&\textbf{5.50}&*{1.00}&\textbf{0.90}&{0.61}&\textbf{0.54}&{0.33}&\textbf{0.29}&{0.16}&\textbf{0.14}&*{0.09}&\textbf{0.08}&*{0.06}&*\textbf{0.05}\\
&\multicolumn{1}{c||}{PPJ}&{39.90}&\textbf{34.95}&{6.49}&\textbf{5.31}&{1.08}&\textbf{0.81}&{0.68}&\textbf{0.51}&{0.39}&\textbf{0.30}&{0.19}&\textbf{0.15}&{0.12}&\textbf{0.10}&{0.09}&\textbf{0.07}\\
\hline
\multirow{4}{*}{\rotatebox[origin=c]{90}{\myinput{livej}}}
&\multicolumn{1}{c||}{ADA}&*{219.4}&*\textbf{192.3}&{64.59}&*\textbf{56.16}&{21.39}&\textbf{17.26}&{14.10}&\textbf{10.77}&{9.46}&\textbf{6.81}&{6.14}&\textbf{4.35}&{3.90}&\textbf{3.18}&{2.86}&\textbf{2.53}\\
&\multicolumn{1}{c||}{ALL}&{365.8}&\textbf{312.6}&{94.86}&\textbf{79.57}&{21.01}&\textbf{16.75}&{10.08}&\textbf{7.90}&{4.83}&\textbf{3.88}&{2.30}&*\textbf{1.92}&*{1.18}&*\textbf{1.08}&*{0.64}&*\textbf{0.62}\\
&\multicolumn{1}{c||}{GRO}&{260.8}&\textbf{\textcolor{gray}{263.3}}&{66.41}&\textbf{66.35}&{16.07}&\textbf{15.91}&{8.40}&\textbf{8.21}&*{4.33}&\textbf{4.21}&*{2.25}&\textbf{2.16}&{1.24}&\textbf{1.21}&{0.72}&\textbf{0.70}\\
&\multicolumn{1}{c||}{PPJ}&{245.8}&\textbf{225.7}&*{64.17}&\textbf{57.97}&*{15.86}&*\textbf{13.91}&*{8.38}&*\textbf{7.30}&{4.38}&*\textbf{3.85}&{2.30}&\textbf{2.04}&{1.23}&\textbf{1.15}&{0.70}&\textbf{0.66}\\
\hline
\multirow{4}{*}{\rotatebox[origin=c]{90}{\myinput{orkut}}}
&\multicolumn{1}{c||}{ADA}&{162.1}&\textbf{\textcolor{gray}{178.9}}&{79.96}&\textbf{\textcolor{gray}{81.73}}&{41.89}&\textbf{40.92}&{29.80}&\textbf{28.72}&{19.71}&\textbf{18.88}&{11.91}&\textbf{11.46}&{7.16}&\textbf{6.87}&{4.27}&\textbf{4.14}\\
&\multicolumn{1}{c||}{ALL}&{220.2}&\textbf{\textcolor{gray}{238.6}}&{70.55}&\textbf{\textcolor{gray}{74.57}}&{25.10}&\textbf{\textcolor{gray}{25.65}}&{15.42}&\textbf{\textcolor{gray}{15.71}}&{9.10}&\textbf{\textcolor{gray}{9.38}}&{5.12}&\textbf{\textcolor{gray}{5.17}}&*{2.81}&*\textbf{\textcolor{gray}{2.84}}&*{1.35}&*\textbf{\textcolor{gray}{1.36}}\\
&\multicolumn{1}{c||}{GRO}&{155.2}&\textbf{\textcolor{gray}{160.7}}&{56.43}&\textbf{\textcolor{gray}{57.34}}&*{22.09}&\textbf{\textcolor{gray}{22.53}}&*{14.00}&\textbf{\textcolor{gray}{14.02}}&{8.71}&\textbf{8.62}&{5.18}&\textbf{5.17}&{3.03}&\textbf{\textcolor{gray}{3.11}}&{1.52}&\textbf{1.51}\\
&\multicolumn{1}{c||}{PPJ}&*{150.1}&*\textbf{\textcolor{gray}{155.4}}&*{54.34}&*\textbf{\textcolor{gray}{55.25}}&{22.14}&*\textbf{21.96}&{14.15}&*\textbf{13.91}&*{8.68}&*\textbf{8.52}&*{5.10}&*\textbf{\textcolor{gray}{5.11}}&{2.93}&\textbf{2.89}&{1.42}&\textbf{1.40}\\
\hline
\multirow{4}{*}{\rotatebox[origin=c]{90}{\myinput{uniform}}}
&\multicolumn{1}{c||}{ADA}&*{56.58}&*\textbf{15.85}&{29.19}&*\textbf{6.80}&{11.03}&*\textbf{2.45}&{7.21}&*\textbf{1.62}&{3.93}&\textbf{0.89}&{1.08}&\textbf{0.28}&{0.54}&\textbf{0.16}&{0.17}&\textbf{0.06}\\
&\multicolumn{1}{c||}{ALL}&{58.78}&\textbf{27.33}&{25.36}&\textbf{14.05}&{8.44}&\textbf{5.30}&{5.38}&\textbf{3.55}&{2.81}&\textbf{1.92}&{0.86}&\textbf{0.58}&{0.44}&\textbf{0.31}&{0.15}&\textbf{0.12}\\
&\multicolumn{1}{c||}{GRO}&{57.30}&\textbf{38.75}&*{19.49}&\textbf{12.39}&*{5.48}&\textbf{3.43}&*{2.95}&\textbf{1.77}&*{1.29}&*\textbf{0.72}&*{0.44}&*\textbf{0.23}&*{0.22}&*\textbf{0.12}&*{0.06}&*\textbf{0.03}\\
&\multicolumn{1}{c||}{PPJ}&{61.14}&\textbf{33.92}&{25.35}&\textbf{15.93}&{8.74}&\textbf{6.14}&{5.59}&\textbf{4.15}&{3.00}&\textbf{2.26}&{1.16}&\textbf{0.86}&{0.65}&\textbf{0.49}&{0.21}&\textbf{0.16}\\
\hline
\multirow{4}{*}{\rotatebox[origin=c]{90}{\myinput{zipf}}}
&\multicolumn{1}{c||}{ADA}&*{1.09}&*\textbf{0.74}&*{0.61}&\textbf{0.45}&{0.40}&\textbf{0.32}&{0.34}&\textbf{0.28}&{0.29}&\textbf{0.23}&{0.23}&\textbf{0.19}&{0.18}&\textbf{0.15}&{0.13}&\textbf{0.12}\\
&\multicolumn{1}{c||}{ALL}&{1.97}&\textbf{0.85}&{0.79}&*\textbf{0.38}&{0.36}&*\textbf{0.20}&{0.25}&*\textbf{0.15}&{0.17}&*\textbf{0.11}&*{0.09}&*\textbf{0.07}&*{0.05}&*\textbf{0.04}&*{0.02}&*\textbf{0.02}\\
&\multicolumn{1}{c||}{GRO}&{2.15}&\textbf{\textcolor{gray}{2.19}}&{0.86}&\textbf{0.85}&{0.39}&\textbf{0.39}&{0.27}&\textbf{0.27}&{0.18}&\textbf{0.17}&{0.11}&\textbf{0.11}&{0.06}&\textbf{0.06}&{0.03}&\textbf{0.03}\\
&\multicolumn{1}{c||}{PPJ}&{1.68}&\textbf{1.06}&{0.70}&\textbf{0.48}&*{0.33}&\textbf{0.26}&*{0.23}&\textbf{0.19}&*{0.16}&\textbf{0.13}&{0.10}&\textbf{0.08}&{0.05}&\textbf{0.05}&{0.03}&\textbf{0.03}\\
\hline
\end{tabular} 
\vspace{0.1cm}
\end{table*}

For each of the 72 different inputs in Table \ref{tb:results_baseline}, there is an asterisk (*) in the algorithm that achieved the best runtime, regarding each group of 4 algorithms (separated for the original and modified algorithm groups). Considering only the original algorithms, it can be seen that \myalgo{GRO} achieved the best runtimes for 33 out of 72 inputs (46\%), followed by \myalgo{ADA} with 15 (21\%), \myalgo{ALL} with 13 (18\%) and \myalgo{PPJ} with 11 (15\%). Considering only the modified algorithms, \myalgo{ADA-BF} and \myalgo{ALL-BF} were the best in 25 out of 72 inputs (35\% each), followed by \myalgo{PPJ-BF} with 12 (17\%) and \myalgo{GRO-BF} with 10 (14\%). Considering both groups together, \myalgo{ADA-BF} and \myalgo{ALL-BF} obtained the best runtime in 23 out of 72 inputs (32\% each), followed by \myalgo{GRO-BF} with 10 (14\%), \myalgo{PPJ-BF} with 9 (13\%), \myalgo{PPJ} with 3 (4\%), \myalgo{ALL} and \myalgo{ADA} with 2 (3\% each). In total, 90\% of the inputs presented best runtime when using algorithms with Bitmap Filter. 

The runtime sum of all 72 experiments for each original algorithm was \myalgo{PPJ}:1535s, \myalgo{GRO}:1561s, \myalgo{ALL}:1878s, \myalgo{ADA}:1944s. Considering the modified algorithms with Bitmap Filter, the runtime sums and the related reduction when compared to the original algorithms were: \myalgo{ADA-BF}:1233s (37\%), \linebreak \myalgo{PPJ-BF}:1359s (11\%), \myalgo{GRO-BF}:1515s (3\%) and \linebreak \myalgo{ALL-BF}:1549s (18\%).

Table \ref{tb:improvement_bitmap64} presents the runtime improvement of the Bitmap Filter considering the formula $(t_{a}/t_{b}-1)\times 100\%$, where $t_{b}$ is the runtime with Bitmap Filter and $t_{a}$ the original baseline runtime. The Bitmap Filter was able to improve the running times in $43$\% on average, although in some situations it improved up to $350$\%. Some experiments presented a small slowdown (negative values) in the computation, but the maximum slowdown was not greater than $9$\%.


\begin{table}
\centering
\caption{Runtime improvement using Bitmap Filter}
\scriptsize
\setlength\tabcolsep{3.0pt}
\begin{tabular}{|@{\hspace{0.1cm}}c@{ }|r||r|r|r|r|r|r|r|r|}
\cline{3-10}
\multicolumn{2}{c|}{} & \multicolumn{8}{c|}{Threshold $\tau_j$}\\
\cline{3-10}
\multicolumn{2}{c|}{} & \parbox{0.66cm}{\centering 0.5} & \parbox{0.66cm}{\centering 0.6} & \parbox{0.66cm}{\centering 0.7} & \parbox{0.66cm}{\centering 0.75} & \parbox{0.66cm}{\centering 0.8} & \parbox{0.66cm}{\centering 0.85} & \parbox{0.66cm}{\centering 0.9} & \parbox{0.66cm}{\centering 0.95} \\
\hline
\multirow{4}{*}{\rotatebox[origin=c]{90}{\myinput{aol}}}
&ADA&{287\%}&{278\%}&{212\%}&{212\%}&{154\%}&{69\%}&{61\%}&{62\%}\\
&ALL&{61\%}&{64\%}&{51\%}&{37\%}&{32\%}&{20\%}&{15\%}&{16\%}\\
&GRO&{6\%}&{7\%}&{13\%}&{14\%}&{17\%}&{18\%}&{18\%}&{18\%}\\
&PPJ&{27\%}&{28\%}&{31\%}&{30\%}&{30\%}&{28\%}&{22\%}&{23\%}\\
\hline
\multirow{4}{*}{\rotatebox[origin=c]{90}{\myinput{bms-pos}}}
&ADA&{179\%}&{186\%}&{166\%}&{166\%}&{155\%}&{130\%}&{108\%}&{94\%}\\
&ALL&{133\%}&{85\%}&{73\%}&{65\%}&{59\%}&{50\%}&{36\%}&{27\%}\\
&GRO&{9\%}&{11\%}&{19\%}&{22\%}&{25\%}&{26\%}&{30\%}&{39\%}\\
&PPJ&{61\%}&{49\%}&{43\%}&{40\%}&{40\%}&{37\%}&{35\%}&{30\%}\\
\hline
\multirow{4}{*}{\rotatebox[origin=c]{90}{\myinput{dblp}}}
&ADA&\textcolor{gray}{-6\%}&{16\%}&{118\%}&{148\%}&{144\%}&{137\%}&{125\%}&{121\%}\\
&ALL&\textcolor{gray}{-7\%}&{12\%}&{39\%}&{51\%}&{58\%}&{69\%}&{75\%}&{63\%}\\
&GRO&\textcolor{gray}{-5\%}&{5\%}&{8\%}&{9\%}&{10\%}&{7\%}&{8\%}&{11\%}\\
&PPJ&\textcolor{gray}{-4\%}&{10\%}&{23\%}&{28\%}&{29\%}&{34\%}&{36\%}&{39\%}\\
\hline
\multirow{4}{*}{\rotatebox[origin=c]{90}{\myinput{enron}}}
&ADA&\textcolor{gray}{-4\%}&{6\%}&{22\%}&{33\%}&{41\%}&{38\%}&{29\%}&{7\%}\\
&ALL&\textcolor{gray}{-4\%}&{7\%}&{21\%}&{25\%}&{21\%}&{16\%}&{8\%}&{1\%}\\
&GRO&\textcolor{gray}{-1\%}&{2\%}&{5\%}&{6\%}&{5\%}&{7\%}&{7\%}&{2\%}\\
&PPJ&{2\%}&{10\%}&{13\%}&{15\%}&{14\%}&{12\%}&{6\%}&{6\%}\\
\hline
\multirow{4}{*}{\rotatebox[origin=c]{90}{\myinput{kosarak}}}
&ADA&{241\%}&{230\%}&{98\%}&{84\%}&{73\%}&{54\%}&{42\%}&{30\%}\\
&ALL&{64\%}&{62\%}&{54\%}&{44\%}&{38\%}&{26\%}&{19\%}&{11\%}\\
&GRO&{5\%}&{6\%}&{11\%}&{12\%}&{11\%}&{7\%}&{12\%}&{12\%}\\
&PPJ&{14\%}&{22\%}&{34\%}&{32\%}&{29\%}&{24\%}&{19\%}&{19\%}\\
\hline
\multirow{4}{*}{\rotatebox[origin=c]{90}{\myinput{livej}}}
&ADA&{14\%}&{15\%}&{24\%}&{31\%}&{39\%}&{41\%}&{23\%}&{13\%}\\
&ALL&{17\%}&{19\%}&{26\%}&{28\%}&{24\%}&{20\%}&{9\%}&{3\%}\\
&GRO&\textcolor{gray}{-1\%}&{0\%}&{1\%}&{2\%}&{3\%}&{4\%}&{2\%}&{3\%}\\
&PPJ&{9\%}&{11\%}&{14\%}&{15\%}&{14\%}&{13\%}&{8\%}&{7\%}\\
\hline
\multirow{4}{*}{\rotatebox[origin=c]{90}{\myinput{orkut}}}
&ADA&\textcolor{gray}{-9\%}&\textcolor{gray}{-2\%}&{2\%}&{4\%}&{4\%}&{4\%}&{4\%}&{3\%}\\
&ALL&\textcolor{gray}{-8\%}&\textcolor{gray}{-5\%}&\textcolor{gray}{-2\%}&\textcolor{gray}{-2\%}&\textcolor{gray}{-3\%}&\textcolor{gray}{-1\%}&\textcolor{gray}{-1\%}&\textcolor{gray}{-1\%}\\
&GRO&\textcolor{gray}{-3\%}&\textcolor{gray}{-2\%}&\textcolor{gray}{-2\%}&\textcolor{gray}{0\%}&{1\%}&{0\%}&\textcolor{gray}{-2\%}&{1\%}\\
&PPJ&\textcolor{gray}{-4\%}&\textcolor{gray}{-2\%}&{1\%}&{2\%}&{2\%}&{0\%}&{2\%}&{2\%}\\
\hline
\multirow{4}{*}{\rotatebox[origin=c]{90}{\myinput{uniform}}}
&ADA&{257\%}&{329\%}&{350\%}&{344\%}&{342\%}&{282\%}&{231\%}&{170\%}\\
&ALL&{115\%}&{81\%}&{59\%}&{52\%}&{47\%}&{47\%}&{41\%}&{28\%}\\
&GRO&{48\%}&{57\%}&{60\%}&{67\%}&{78\%}&{92\%}&{81\%}&{80\%}\\
&PPJ&{80\%}&{59\%}&{42\%}&{35\%}&{33\%}&{35\%}&{34\%}&{30\%}\\
\hline
\multirow{4}{*}{\rotatebox[origin=c]{90}{\myinput{zipf}}}
&ADA&{48\%}&{35\%}&{24\%}&{25\%}&{24\%}&{24\%}&{20\%}&{11\%}\\
&ALL&{130\%}&{108\%}&{78\%}&{69\%}&{54\%}&{30\%}&{19\%}&{0\%}\\
&GRO&\textcolor{gray}{-2\%}&{1\%}&{2\%}&{0\%}&{6\%}&{2\%}&{4\%}&{0\%}\\
&PPJ&{59\%}&{45\%}&{30\%}&{21\%}&{19\%}&{16\%}&{8\%}&{0\%}\\
\hline
\end{tabular} 
\label{tb:improvement_bitmap64}
\end{table}

Table \ref{tb:improvement_algorithms} shows the runtime improvement of the algorithms with respect to each threshold, considering the runtime average of all collections. All the algorithms presented improvements: \myalgo{ADA-BF} was the one with highest gain, varying from 57\% up to 121\%; \myalgo{ALL-BF} presented improvements ranging from 17\% to 56\%; \myalgo{PPJ-BF} showed improvements from 17\% to 27\%; \myalgo{GRO-BF} was the one with lowest gains, varying from 6\% up to 18\%. 

\begin{table}
\centering
\caption{Average Bitmap Filter improvement per algorithm}
\scriptsize
\setlength\tabcolsep{3.0pt}
\begin{tabular}{|@{\hspace{0.1cm}}c@{ }|r||r|r|r|r|r|r|r|r|}
\cline{3-10}
\multicolumn{2}{c|}{} & \multicolumn{8}{c|}{Threshold $\tau_j$}\\
\cline{3-10}
\multicolumn{2}{c|}{} & \parbox{0.65cm}{\centering 0.5} & \parbox{0.65cm}{\centering 0.6} & \parbox{0.65cm}{\centering 0.7} & \parbox{0.65cm}{\centering 0.75} & \parbox{0.65cm}{\centering 0.8} & \parbox{0.65cm}{\centering 0.85} & \parbox{0.65cm}{\centering 0.9} & \parbox{0.65cm}{\centering 0.95} \\
\hline
\multirow{4}{*}{\rotatebox[origin=c]{90}{Average}}
&ADA&{112\%}&{121\%}&{113\%}&{116\%}&{109\%}&{86\%}&{72\%}&{57\%}\\
&ALL&{56\%}&{48\%}&{44\%}&{41\%}&{37\%}&{31\%}&{25\%}&{17\%}\\
&GRO&{6\%}&{10\%}&{13\%}&{15\%}&{17\%}&{18\%}&{18\%}&{18\%}\\
&PPJ&{27\%}&{26\%}&{26\%}&{24\%}&{23\%}&{22\%}&{19\%}&{17\%}\\
\hline
\end{tabular} 
\label{tb:improvement_algorithms}
\end{table}

Table \ref{tb:improvement_collections} presents the average improvement for each collection and threshold, considering the average runtimes between the algorithms. A very good improvement is noted for collections \myinput{uniform}, \myinput{bms-pos}, \myinput{aol}, \myinput{dblp} and \myinput{kosarak}, followed by intermediate gains in \myinput{livej} and \myinput{zip}. The smallest improvements were observed in collections \myinput{enron} and \myinput{orkut}. Small slowdowns of up to 6\% were detected in collections \myinput{dblp}, \myinput{enron} and \myinput{orkut} at low thresholds (0.5 and 0.6). The collection gain is related to the set size distribution (Figure~\ref{fig:setsizes}) and the number of distinct tokens (Table \ref{tb:collections}), such that collections with many small sets and few unique tokens (e.g. \myinput{uniform} and \myinput{bms-pos}) present higher improvements then collections with many large sets and many unique tokens (e.g. \myinput{enron} and \myinput{orkut}).


\begin{table}
\caption{Average Bitmap Filter improvement per collection}
\label{tb:improvement_collections}
\scriptsize
\setlength\tabcolsep{3.0pt}
\begin{tabular}{|@{\hspace{0.1cm}}p{1.1cm}@{ }||r|r|r|r|r|r|r|r|}
\cline{2-9}
\multicolumn{1}{c|}{} & \multicolumn{8}{c|}{Threshold $\tau_j$}\\
\cline{2-9}
\multicolumn{1}{c|}{} & \parbox{0.65cm}{\centering 0.5} & \parbox{0.65cm}{\centering 0.6} & \parbox{0.65cm}{\centering 0.7} & \parbox{0.65cm}{\centering 0.75} & \parbox{0.65cm}{\centering 0.8} & \parbox{0.65cm}{\centering 0.85} & \parbox{0.65cm}{\centering 0.9} & \parbox{0.65cm}{\centering 0.95} \\
\hline
\myinput{aol}&{95\%}&{94\%}&{77\%}&{73\%}&{58\%}&{34\%}&{29\%}&{29\%}\\
\myinput{bms-pos}&{96\%}&{83\%}&{75\%}&{73\%}&{70\%}&{61\%}&{53\%}&{48\%}\\
\myinput{dblp}&\textcolor{gray}{-6\%}&{11\%}&{47\%}&{59\%}&{60\%}&{61\%}&{61\%}&{58\%}\\
\myinput{enron}&\textcolor{gray}{-2\%}&{6\%}&{15\%}&{20\%}&{20\%}&{18\%}&{13\%}&{4\%}\\
\myinput{kosarak}&{81\%}&{80\%}&{49\%}&{43\%}&{38\%}&{28\%}&{23\%}&{18\%}\\
\myinput{livej}&{10\%}&{11\%}&{16\%}&{19\%}&{20\%}&{19\%}&{10\%}&{6\%}\\
\myinput{orkut}&\textcolor{gray}{-6\%}&\textcolor{gray}{-3\%}&{0\%}&{1\%}&{1\%}&{1\%}&{1\%}&{1\%}\\
\myinput{uniform}&{125\%}&{132\%}&{128\%}&{124\%}&{125\%}&{114\%}&{97\%}&{77\%}\\
\myinput{zipf}&{59\%}&{47\%}&{33\%}&{29\%}&{26\%}&{18\%}&{13\%}&{3\%}\\
\hline
\end{tabular} 
\end{table}

\subsubsection{Filtering Ratio}

In order to assess the efficiency of the Bitmap Filter, we collected the number of candidate pairs that were filtered out by the Bitmap Filter. Table \ref{tb:filter_ratio} presents the filtering ratio, defined by the number of filtered pairs divided by the total number of candidates. The filtered ratio is strongly related to the runtime improvement (Table \ref{tb:improvement_bitmap64}).

\begin{table}
\centering
\caption{Bitmap Filter ratio per collection (AllPairs Algorithm)}
\label{tb:filter_ratio}
\scriptsize
\setlength\tabcolsep{3.0pt}
\begin{tabular}{|@{\hspace{0.1cm}}p{1.1cm}@{ }||r|r|r|r|r|r|r|r|}
\cline{2-9}
\multicolumn{1}{c|}{} & \multicolumn{8}{c|}{Threshold $\tau_j$}\\
\cline{2-9}
\multicolumn{1}{c|}{} & \parbox{0.65cm}{\centering 0.5} & \parbox{0.65cm}{\centering 0.6} & \parbox{0.65cm}{\centering 0.7} & \parbox{0.65cm}{\centering 0.75} & \parbox{0.65cm}{\centering 0.8} & \parbox{0.65cm}{\centering 0.85} & \parbox{0.65cm}{\centering 0.9} & \parbox{0.65cm}{\centering 0.95} \\
\hline
\myinput{aol}&{97\%}&{98\%}&{99\%}&{99\%}&{99\%}&{99\%}&{99\%}&{99\%}\\
\myinput{bms-pos}&{98\%}&{99\%}&{99\%}&{99\%}&{99\%}&{99\%}&{99\%}&{99\%}\\
\myinput{dblp}&{11\%}&{54\%}&{97\%}&{99\%}&{99\%}&{99\%}&{99\%}&{99\%}\\
\myinput{enron}&{14\%}&{29\%}&{56\%}&{71\%}&{83\%}&{86\%}&{89\%}&{85\%}\\
\myinput{kosarak}&{86\%}&{90\%}&{93\%}&{95\%}&{96\%}&{98\%}&{99\%}&{99\%}\\
\myinput{livej}&{45\%}&{52\%}&{64\%}&{74\%}&{86\%}&{99\%}&{99\%}&{99\%}\\
\myinput{orkut}&{4\%}&{5\%}&{8\%}&{10\%}&{13\%}&{18\%}&{27\%}&{54\%}\\
\myinput{uniform}&{99\%}&{99\%}&{99\%}&{99\%}&{99\%}&{99\%}&{99\%}&{99\%}\\
\myinput{zipf}&{99\%}&{99\%}&{100\%}&{100\%}&{100\%}&{100\%}&{100\%}&{100\%}\\
\hline
\end{tabular} 
\end{table}

Figure \ref{fig:filter_ratio:bitmap_types} presents the filtering ratio for the three bitmap creation algorithms, using bitmap sizes with $b=64$ bits and without the cutoff point (Section \ref{sec:bitmap:cutoff}). In the figure, the Bitmap-Xor was consistently the best option for the three analyzed collections, followed by Bitmap-Set and Bitmap-Next. The only situation where the Bitmap-Set was the best option was for the \myinput{enron} collection with Jaccard threshold 0.5. As stated in Section \ref{sec:bitmap:cutoff}, Bitmap-Set is slightly better around $\tau_j=0.5$. In the \myinput{dblp} collection, the Bitmap-Xor presented filtering ratios up to 6.37$\times$ better than Bitmap-Set, due to the large set sizes presented in \myinput{dblp}. This high filtering ratio difference reduced the runtime by up to $55\%$.

\begin{figure}[t]
\includegraphics[width=8.5cm]{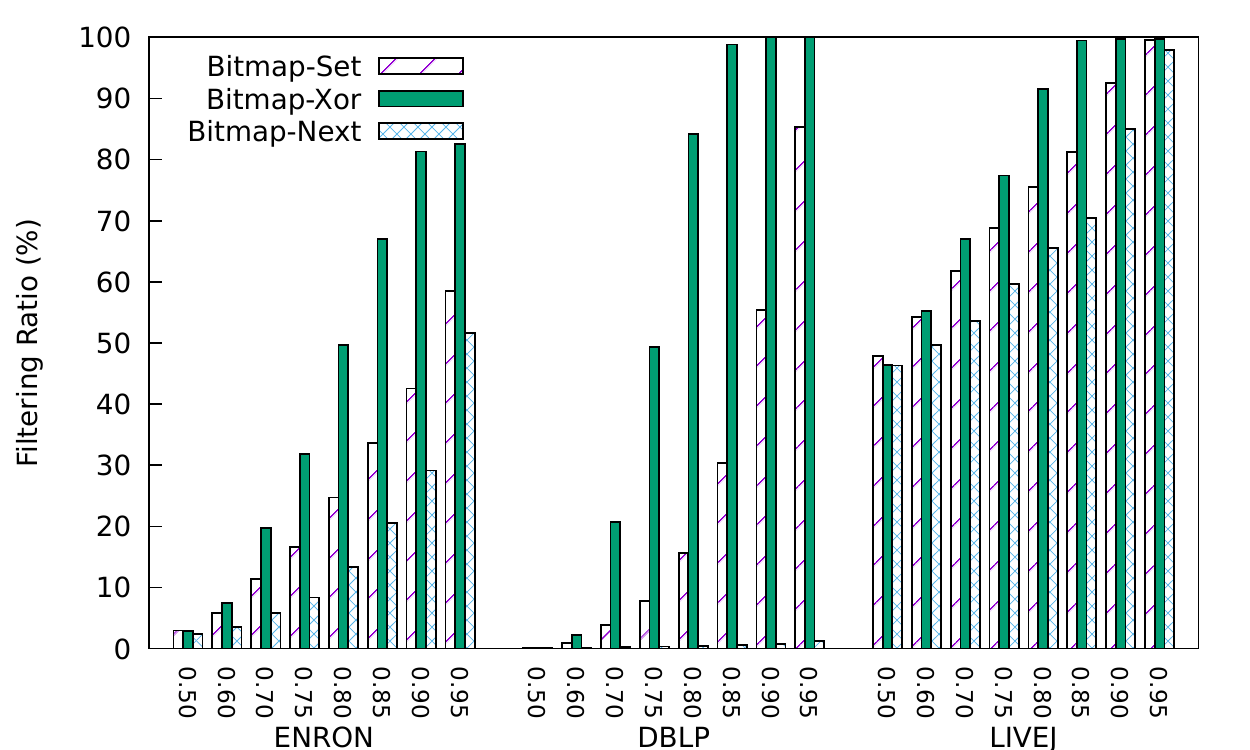}
\caption{Bitmap Filter ratio for different Bitmap generation methods and Jaccard thresholds $\tau_j$ ($b=64$ bits)\label{fig:filter_ratio:bitmap_types}}
\end{figure}

\subsubsection{Filtering Precision}

The Bitmap Filter can be considered as a binary classification rule that classifies a candidate pair as \textit{unfiltered} (positive class) or \textit{filtered} (negative class). Since the Bitmap Filter is an exact method, there is no wrongly filtered pair (false negative), so the Bitmap Filter has a $1.0$ recall (true positives divided by true positives plus false negatives). Nevertheless, there may be some dissimilar pairs that are not filtered (false positive), leading to an extra verification cost. We define the filtering precision as the ratio between similar pairs (true positives) and unfiltered pairs (false positives plus true positives). The number of false positives tends to increase when the bitmaps are too small or when the number or unique tokens are very large. 

The cutoff point defined in Section \ref{sec:bitmap:cutoff} disables the Bitmap Filter whenever the filtering precision falls drastically. In order to show this precision drop off, the filtering precision was measured for the self-joins of \myinput{dblp} and \myinput{enron} collections, without using the cutoff point. The bitmap size used for this experiment was $b=64$. Figure \ref{fig:filterprecision} shows the average filtering precision for the different set sizes found in the collection, considering Jaccard threshold $\tau_j$ varying from 0.50 to 0.80. Although each collection has its own characteristics, we can see the precision drop off in almost the same position in the plots. This clearly represents the cutoff point, such that beyond this point the filter becomes inefficient. So, higher runtime improvements tends to occur in collections whose majority of sets contains less tokens than the observed precision drop off point.

\begin{figure}[t]
\includegraphics[width=8.5cm]{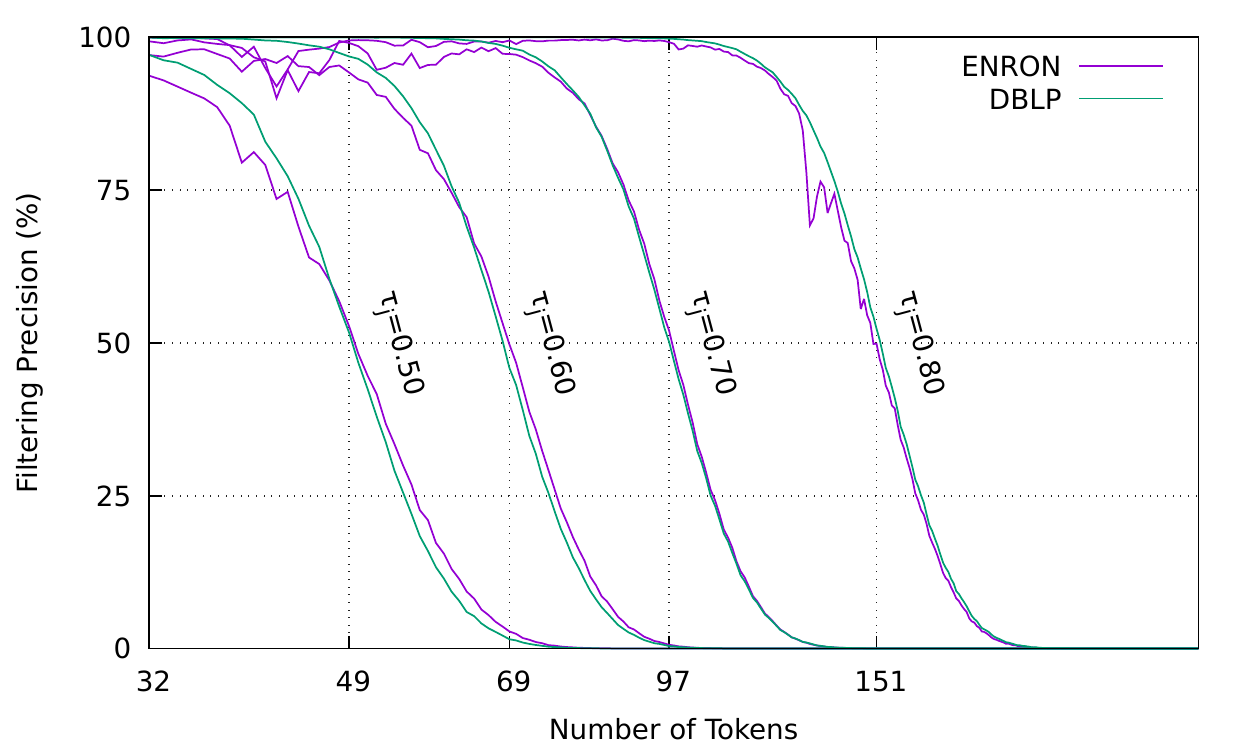}
\caption{Bitmap Filter precision for different number of tokens and Jaccard thresholds $\tau_j$ ($b=64$ bits)\label{fig:filterprecision}}
\end{figure}

\subsection{GPU Experiments}

The GPU algorithm was executed over the collections of up to a 606,770 sets: \myinput{bms-pos}, \myinput{dblp}, \myinput{enron}, \myinput{kosarak}, \myinput{uniform}, \myinput{zipf}. The algorithm was tested with bitmap sizes ranging from $b=64$ to $4096$ and threshold $\tau_j$ from $0.5$ to $0.75$. The number of CUDA blocks was fixed in 512, each one with 256 threads. Table \ref{tb:results_gpu} presents the running times of the GPU algorithm considering the bitmap of size $b$ which presented the best execution time. The running times does not include the file loading time and device initialization, but includes the GPU memory allocation, data transfers from/to the GPU, the GPU kernel execution and the CPU procedures related to the join algorithm. Table \ref{tb:results_gpu} also gives the speedup up of the GPU implementation compared to the best runtime obtained in the experiment with the baseline CPU algorithms (Table \ref{tb:results_baseline}). 


\begin{table}
\scriptsize
\centering
\caption{Runtimes (in seconds) and speedup of the GPU algorithm}
\label{tb:results_gpu}
\setlength\tabcolsep{3.0pt}
\begin{tabular}{|c|r||r|r|r|r|}
\cline{3-6}
\multicolumn{2}{c|}{} & \multicolumn{4}{c|}{Threshold $\tau_j$}\\
\cline{3-6}
\multicolumn{2}{c|}{} & \parbox{1.35cm}{\centering 0.5} & \parbox{1.35cm}{\centering 0.6} & \parbox{1.35cm}{\centering 0.7} & \parbox{1.35cm}{\centering 0.75} \\
\hline
\multirow{5}{*}{\rotatebox[origin=c]{90}{\myinput{bms-pos}}}
& GPU time
&{0.77}&{0.41}&{0.31}&{0.28}\\
& CPU time & (\myalgo{ADA})24.49 & (\myalgo{ADA})8.47 & (\myalgo{ADA})2.61 & (\myalgo{GRO})1.62\\
& speedup
&{31.7$\times$}&{20.6$\times$}&{8.3$\times$}&{5.7$\times$}\\
&bitset $b$
&192bits&128bits&128bits&128bits\\

\hline
\multirow{5}{*}{\rotatebox[origin=c]{90}{\myinput{dblp}}}
& GPU time
&{0.28}&{0.18}&{0.12}&{0.09}\\
& CPU time & (\myalgo{ADA})161.59 & (\myalgo{ADA})75.25 & (\myalgo{ADA})31.01 & (\myalgo{ADA})17.47\\
& speedup
&{577.9$\times$}&{414.8$\times$}&{250.5$\times$}&{190.3$\times$}\\
&bitset $b$
&512bits&320bits&192bits&192bits\\

\hline
\multirow{5}{*}{\rotatebox[origin=c]{90}{\myinput{enron}}}
& GPU time
&{3.27}&{1.92}&{1.26}&{1.00}\\
& CPU time & (\myalgo{ADA})30.35 & (\myalgo{GRO})9.90 & (\myalgo{GRO})2.93 & (\myalgo{GRO})1.69\\
& speedup
&{9.3$\times$}&{5.2$\times$}&{2.3$\times$}&{1.7$\times$}\\
&bitset $b$
&2560bits&2560bits&2048bits&1280bits\\

\hline
\multirow{5}{*}{\rotatebox[origin=c]{90}{\myinput{kosarak}}}
& GPU time
&{7.65}&{2.64}&{1.21}&{0.95}\\
& CPU time & (\myalgo{GRO})37.54 & (\myalgo{GRO})5.82 & (\myalgo{GRO})1.00 & (\myalgo{ALL})0.59\\
& speedup
&{4.9$\times$}&{2.2$\times$}&\textcolor{gray}{0.8$\times$}&\textcolor{gray}{0.6$\times$}\\
&bitset $b$
&512bits&256bits&192bits&128bits\\

\hline
\multirow{5}{*}{\rotatebox[origin=c]{90}{\myinput{uniform}}}
& GPU time
&{1.96}&{0.19}&{0.07}&{0.06}\\
& CPU time & (\myalgo{ADA})56.58 & (\myalgo{GRO})19.49 & (\myalgo{GRO})5.48 & (\myalgo{GRO})2.95\\
& speedup
&{28.9$\times$}&{100.5$\times$}&{79.9$\times$}&{47.8$\times$}\\
&bitset $b$
&192bits&128bits&128bits&128bits\\

\hline
\multirow{4}{*}{\rotatebox[origin=c]{90}{\myinput{zipf}}}
& GPU time
&{0.14}&{0.11}&{0.09}&{0.09}\\
& CPU time & (\myalgo{ADA})1.09 & (\myalgo{ADA})0.61 & (\myalgo{PPJ})0.33 & (\myalgo{PPJ})0.23\\
& speedup
&{7.8$\times$}&{5.8$\times$}&{3.5$\times$}&{2.6$\times$}\\
&bitset $b$
&192bits&128bits&128bits&128bits\\
\hline
\end{tabular} 
\end{table}

It is worth mentioning that the GPU algorithm applies the Bitmap Filter for every possible pair of sets that were not pruned by the Length Filter. Although the algorithm did not include the Prefix Filter, the bitwise operations in the bitmap comparisons are extremely fast in GPU devices. \linebreak Thus, the GPUs are able to process a huge number of bitmap comparisons in short time, with bitmap sizes up to 4096 bits. The best results were noticed in the \myinput{dblp} collection, with speedups of up to 577.9$\times$, followed by \myinput{uniform} (100.5$\times$) and \myinput{bms-pos} (31.7$\times$).

\section{Conclusions}
\label{sec:conclusion}

\begin{sloppypar}
This paper presented a new filtering technique called Bitmap Filter, which is able to reduce the running time of state-of-the-art algorithms that solve the exact Set Similarity Join problem. The Bitmap Filter uses hash functions to create binary words of $b$ bits to represent the set of tokens, with reduced dimension. The comparison of two bitmaps using population count and \textit{xor} operations (hamming distance) allows to infer the number of different tokens in the original sets, and this difference is mapped to an upper bound for the overlap between these sets. 
\end{sloppypar}

Three bitmap generation methods were proposed for generating the bitmaps: Bitmap-Set, Bitmap-Xor and Bitmap-Next. Each method presents better filtering ratio at some threshold $\tau$. The Bitmap-Combined method selects the best method according to the given threshold. 

The Bitmap Filter was implemented for CPU in four state-of-the art algorithm for exact Set Similarity Join: AllPairs, PPJoin, AdaptJoin and GroupJoin. In the experimental results, we showed that the Bitmap Filter was able to speedup 90\% of the experiments, improving the runtime up to 350\% and 43\% on average. Some slowdown was observed in few cases, but restricted to less than 10\%. Using the equations produced in this paper, it is possible to identify the maximum set size such that the Bitmap Filter will be able to filter with good precision. Beyond this cutoff point the Bitmap Filter may be disabled.

The core of the Bitmap Filter is implemented with bitwise operations, allowing it to be extremely fast in GPUs. The paper also presented a GPU implementation where the Bitmap Filter is executed for every possible pair of sets that are not pruned by the Length Filter. Although this GPU implementation is not supposed to scale for large collections, it was still capable to speedup the join operation on collections with up to 606,770 sets. Using an Nvidia Geforce GTX 980 Ti GPU, the experimental results showed speedups of up to 577$\times$ for a \myinput{dblp} sample with 100,000 sets and Jaccard threshold of $0.5$. 

We believe that the simplicity of the Bitmap Filter will allow the development of many other algorithms for the exact set similarity join problem. As future work, we intend to implement the Bitmap Filter in different GPU algorithms using prefix and positional filters. We also plan to assess the Bitmap Filter speedups in other platforms such as Intel Phi, AMD GPU boards, as well as using CPU vector instructions set such as Advanced Vector Extensions (AVX).

\balance


{
\bibliographystyle{abbrv}
\bibliography{bitmap_filter}  

\begin{thebibliography}{10}

\bibitem{PartEnum:2006}
A.~Arasu, V.~Ganti, and R.~Kaushik.
\newblock Efficient exact set-similarity joins.
\newblock In {\em Proceedings of the 32Nd International Conference on Very
  Large Data Bases}, VLDB '06, pages 918--929. VLDB Endowment, 2006.

\bibitem{Augsten:2013:SJR:2601748}
N.~Augsten and M.~H. Bhlen.
\newblock {\em Similarity Joins in Relational Database Systems}.
\newblock Morgan \& Claypool Publishers, 1st edition, 2013.

\bibitem{AllPairs:2007}
R.~J. Bayardo, Y.~Ma, and R.~Srikant.
\newblock Scaling up all pairs similarity search.
\newblock In {\em Proceedings of the 16th International Conference on World
  Wide Web}, WWW '07, pages 131--140, New York, NY, USA, 2007. ACM.

\bibitem{GroupJoin:2012}
P.~Bouros, S.~Ge, and N.~Mamoulis.
\newblock Spatio-textual similarity joins.
\newblock {\em Proc. VLDB Endow.}, 6(1):1--12, Nov. 2012.

\bibitem{Broder:1997:Clustering}
A.~Z. Broder, S.~C. Glassman, M.~S. Manasse, and G.~Zweig.
\newblock Syntactic clustering of the web.
\newblock {\em Computer Networks and ISDN Systems}, 29(8-13):1157--1166, Sept.
  1997.

\bibitem{Chakrabarti:2015:LSH}
A.~Chakrabarti and S.~Parthasarathy.
\newblock Sequential hypothesis tests for adaptive locality sensitive hashing.
\newblock In {\em Proceedings of the 24th International Conference on World
  Wide Web}, WWW '15, pages 162--172, Republic and Canton of Geneva,
  Switzerland, 2015. International World Wide Web Conferences Steering
  Committee.

\bibitem{Chaudhuri:2006}
S.~Chaudhuri, V.~Ganti, and R.~Kaushik.
\newblock A primitive operator for similarity joins in data cleaning.
\newblock In {\em Proceedings of the 22Nd International Conference on Data
  Engineering}, ICDE '06, pages 5--, Washington, DC, USA, 2006. IEEE Computer
  Society.

\bibitem{Cohen:2000:DIU:352595.352598}
W.~W. Cohen.
\newblock Data integration using similarity joins and a word-based information
  representation language.
\newblock {\em ACM Trans. Inf. Syst.}, 18(3):288--321, July 2000.

\bibitem{Gionis:1999:LSH}
A.~Gionis, P.~Indyk, and R.~Motwani.
\newblock Similarity search in high dimensions via hashing.
\newblock In {\em Proceedings of the 25th International Conference on Very
  Large Data Bases}, VLDB '99, pages 518--529, San Francisco, CA, USA, 1999.
  Morgan Kaufmann Publishers Inc.

\bibitem{GramCount:2001}
L.~Gravano, P.~G. Ipeirotis, H.~V. Jagadish, N.~Koudas, S.~Muthukrishnan, and
  D.~Srivastava.
\newblock Approximate string joins in a database (almost) for free.
\newblock In {\em Proceedings of the 27th International Conference on Very
  Large Data Bases}, VLDB '01, pages 491--500, San Francisco, CA, USA, 2001.
  Morgan Kaufmann Publishers Inc.

\bibitem{LSH:1998}
P.~Indyk and R.~Motwani.
\newblock Approximate nearest neighbors: Towards removing the curse of
  dimensionality.
\newblock In {\em Proceedings of the Thirtieth Annual ACM Symposium on Theory
  of Computing}, STOC '98, pages 604--613, New York, NY, USA, 1998. ACM.

\bibitem{Jiang:2014}
Y.~Jiang, G.~Li, J.~Feng, and W.-S. Li.
\newblock String similarity joins: An experimental evaluation.
\newblock {\em Proc. VLDB Endow.}, 7(8):625--636, Apr. 2014.

\bibitem{mann2016}
W.~Mann, N.~Augsten, and P.~Bouros.
\newblock An empirical evaluation of set similarity join techniques.
\newblock {\em Proc. VLDB Endow.}, 9(9):636--647, May 2016.

\bibitem{CUDAGuide8.0}
{Nvidia}.
\newblock {\em Nvidia CUDA Programming Guide 8.0}.
\newblock {Nvidia}, 2017.

\bibitem{PEYRAVIAN1998171}
M.~Peyravian, A.~Roginsky, and A.~Kshemkalyani.
\newblock On probabilities of hash value matches.
\newblock {\em Computers \& Security}, 17(2):171 -- 176, 1998.

\bibitem{Ribeiro:2011}
L.~A. Ribeiro and T.~H\"{a}rder.
\newblock Generalizing prefix filtering to improve set similarity joins.
\newblock {\em Inf. Syst.}, 36(1):62--78, Mar. 2011.

\bibitem{Satuluri:2012:LSH}
V.~Satuluri and S.~Parthasarathy.
\newblock Bayesian locality sensitive hashing for fast similarity search.
\newblock {\em Proc. VLDB Endow.}, 5(5):430--441, Jan. 2012.

\bibitem{nehalem2008inside}
R.~Singhal.
\newblock {Inside Intel{\textregistered} Next Generation Nehalem
  Microarchitecture}.
\newblock In {\em Hot Chips}, volume~20, page~15, 2008.

\bibitem{SOHRABI20171:LSH}
M.~K. Sohrabi and H.~Azgomi.
\newblock Parallel set similarity join on big data based on locality-sensitive
  hashing.
\newblock {\em Science of Computer Programming}, 145(Supplement C):1 -- 12,
  2017.

\bibitem{Spertus:2005:ESM:1081870.1081956}
E.~Spertus, M.~Sahami, and O.~Buyukkokten.
\newblock Evaluating similarity measures: A large-scale study in the orkut
  social network.
\newblock In {\em Proceedings of the Eleventh ACM SIGKDD International
  Conference on Knowledge Discovery in Data Mining}, KDD '05, pages 678--684,
  New York, NY, USA, 2005. ACM.

\bibitem{Theobald:2008:SRE:1390334.1390431}
M.~Theobald, J.~Siddharth, and A.~Paepcke.
\newblock Spotsigs: Robust and efficient near duplicate detection in large web
  collections.
\newblock In {\em Proceedings of the 31st Annual International ACM SIGIR
  Conference on Research and Development in Information Retrieval}, SIGIR '08,
  pages 563--570, New York, NY, USA, 2008. ACM.

\bibitem{Vernica:2010:MapReduce}
R.~Vernica, M.~J. Carey, and C.~Li.
\newblock Efficient parallel set-similarity joins using mapreduce.
\newblock In {\em Proceedings of the 2010 ACM SIGMOD International Conference
  on Management of Data}, SIGMOD '10, pages 495--506, New York, NY, USA, 2010.
  ACM.

\bibitem{AdaptJoin:2012}
J.~Wang, G.~Li, and J.~Feng.
\newblock Can we beat the prefix filtering?: An adaptive framework for
  similarity join and search.
\newblock In {\em Proceedings of the 2012 ACM SIGMOD International Conference
  on Management of Data}, SIGMOD '12, pages 85--96, New York, NY, USA, 2012.
  ACM.

\bibitem{Wang:2017:LSR:3099622.3099624}
X.~Wang, L.~Qin, X.~Lin, Y.~Zhang, and L.~Chang.
\newblock Leveraging set relations in exact set similarity join.
\newblock {\em Proc. VLDB Endow.}, 10(9):925--936, May 2017.

\bibitem{PPJoin:2011}
C.~Xiao, W.~Wang, X.~Lin, J.~X. Yu, and G.~Wang.
\newblock Efficient similarity joins for near-duplicate detection.
\newblock {\em ACM Trans. Database Syst.}, 36(3):15:1--15:41, Aug. 2011.

\bibitem{Yu:2017:LSH}
C.~Yu, S.~Nutanong, H.~Li, C.~Wang, and X.~Yuan.
\newblock A generic method for accelerating lsh-based similarity join
  processing.
\newblock {\em IEEE Transactions on Knowledge and Data Engineering},
  29(4):712--726, April 2017.

\bibitem{Zhai:2011:LSH}
J.~Zhai, Y.~Lou, and J.~Gehrke.
\newblock Atlas: A probabilistic algorithm for high dimensional similarity
  search.
\newblock In {\em Proceedings of the 2011 ACM SIGMOD International Conference
  on Management of Data}, SIGMOD '11, pages 997--1008, New York, NY, USA, 2011.
  ACM.

\bibitem{Zhang:2017:Exact}
Y.~Zhang, X.~Li, J.~Wang, Y.~Zhang, C.~Xing, and X.~Yuan.
\newblock An efficient framework for exact set similarity search using tree
  structure indexes.
\newblock In {\em 2017 IEEE 33rd International Conference on Data Engineering
  (ICDE)}, pages 759--770, April 2017.

\end{thebibliography}
}



%
%
%
%

\end{document}